\documentclass[a4paper,USenglish,cleveref, thm-restate, numberwithinsect]{lipics-v2021}

\hideLIPIcs  

\title{Search-space Reduction for Boolean MinCSPs via Essential Constraints}

\titlerunning{Search-space Reduction for Boolean MinCSPs via Essential Constraints}

\author{Bart M.\,P. Jansen}{Eindhoven University of Technology, The Netherlands}{b.m.p.jansen@tue.nl}{https://orcid.org/0000-0001-8204-1268}{} 

\author{Ruben F.\,A. Verhaegh}{Eindhoven University of Technology, The Netherlands}{r.f.a.verhaegh@tue.nl}{https://orcid.org/0009-0008-8568-104X}{}

\authorrunning{B.M.P.~Jansen and R.F.A.~Verhaegh}

\Copyright{Bart M.P.~Jansen and Ruben F.A.~Verhaegh}

\ccsdesc[300]{Theory of computation~Graph algorithms analysis}
\ccsdesc[500]{Theory of computation~Approximation algorithms analysis}
\ccsdesc[500]{Theory of computation~Parameterized complexity and exact algorithms}
\ccsdesc[500]{Theory of computation~Constraint and logic programming}

\keywords{fixed-parameter tractability, constraint satisfaction problems}

\category{} 

\relatedversion{} 

\nolinenumbers 

\EventEditors{Pierre Fraigniaud}
\EventNoEds{1}
\EventLongTitle{20th Scandinavian Symposium on Algorithm Theory (SWAT 2026)}
\EventShortTitle{SWAT 2026}
\EventAcronym{SWAT}
\EventYear{2026}
\EventDate{June 17--19, 2026}
\EventLocation{Copenhagen, Denmark}
\EventLogo{}
\SeriesVolume{370}
\ArticleNo{29}
\newcommand{\problem}[3]{
  \vspace{1mm}
  \noindent\fbox{
    \begin{minipage}{0.96\textwidth}
      \begin{tabularx}{\textwidth}{@{\hspace{\parindent}}l X}
        \multicolumn{2}{@{\hspace{\parindent}}l}{\textsc{#1}} \\
        \textbf{Input:} & #2 \\
        \textbf{Task:} & #3 \\
      \end{tabularx}
    \end{minipage}
  }
  \vspace{1mm}
}

\crefname{claim}{Claim}{Claims}
\crefname{lemma}{Lemma}{Lemmas}

\newcommand{\Oh}{\mathcal{O}}
\newcommand{\eps}{\varepsilon}
\newcommand{\F}{\mathcal{F}}
\newcommand{\X}{\mathbf{X}}
\newcommand{\Y}{\mathbf{Y}}
\newcommand{\Z}{\mathbf{Z}}
\newcommand{\C}{\mathcal{C}}
\newcommand{\R}{\mathbb{R}}
\newcommand{\Rone}{\R_{\geq 1}}
\newcommand{\fneq}{\F_{\mathrm{\neq}}}

\usepackage{xspace}
\usepackage{makecell}
\newcommand{\csp}{\textsc{MinCSP}}
\newcommand{\mincsp}[1]{\csp(#1)\xspace}
\newcommand{\minF}{\mincsp{$\F$}}
\newcommand{\wmincsp}[1]{\textsc{Weighted} \mincsp{#1}}
\newcommand{\wminF}{\wmincsp{$\F$}}
\newcommand{\opt}{\mathrm{opt}}
\newcommand{\inv}{^{-1}}

\newcommand{\imp}{\xrightarrow{\mathrm{ess.}}}

\newcommand{\stack}[2]{\begin{matrix} #1 \\ #2 \end{matrix}}
\newcommand{\comma}{\stack{}{,}~}

\newcommand{\true}{\mathrm{\textsc{True}}}
\newcommand{\false}{\mathrm{\textsc{False}}}
\newcommand{\nae}{\mathrm{NAE}}
\newcommand{\orr}{\mathrm{OR}}
\newcommand{\nor}{\mathrm{NOR}}
\newcommand{\xor}{\mathrm{XOR}}
\newcommand{\xnor}{\mathrm{XNOR}}
\newcommand{\neqq}{[\neq]}
\newcommand{\eqq}{[=]}
\newcommand{\horn}{\mathrm{HORN}}

\begin{document}

\maketitle
\begin{abstract}
For a fixed set~$\F$ of Boolean constraint types, a \mincsp{$\F$}-instance consists of a formula~$F$ that applies~$m$ constraints from~$\F$ to a set of~$n$ Boolean variables. The goal is to remove a minimum subset of constraint applications from~$F$ to make the remaining formula satisfiable. 
Previous work characterized how the choice of~$\F$ affects its polynomial-time solvability and approximability. 
We extend a recently introduced preprocessing framework for graph problems to the problem above. Rephrased in the context of CSPs, this framework defines a constraint application from a given formula~$F$ as~\emph{$c$-essential} if it is contained in all~$c$-approximate solutions to~$F$. Being able to efficiently detect these essential parts of a solution reduces the search space of any follow-up FPT algorithms parameterized by the solution size and yields an immediate asymptotic improvement to the runtime of such algorithms. In this work, we present a dichotomy theorem that distinguishes constraint sets~$\F$ for which $c_\mathcal{F}$-essential constraint applications can be detected efficiently for some~$c_{\mathcal{F}} \in \Oh(1)$, from those for which this task is intractable under established complexity-theoretic conjectures. Our results show that for any set~$\F$ of \emph{bijunctive} constraints, there is a polynomial-time algorithm that detects~$\Oh(1)$-essential constraint applications. 
This contrasts the fact that constant-factor approximating a bijunctive \minF-problem is intractable under the Unique Games Conjecture.
\end{abstract}

\clearpage

\section{Introduction}
\label{sec:introduction}
Constraint Satisfaction Problems (CSPs) form a broad class of problems that describe many well-known computational tasks and allow us to study them in a unified language~\cite{DFUVol7}. A \emph{constraint} of arity~$r$ over a domain~$D$ is a function~$f \colon D^r \to \{0,1\}$, and a \emph{constraint set}~$\F$ is a set of such constraints over the same domain and possibly of different arity. In this work, we focus on Boolean CSPs, and assume the domain~$D$ to be~$\{0,1\}$ from now on. For a fixed constraint set~$\F$, a CSP instance contains a formula consisting of multiple applications of constraints from~$\F$ over a set of variables~$\X$. Such a \emph{constraint application} is a pair consisting of a constraint~$f \in \F$ and a sequence~$\X'$, whose length equals the arity of~$f$, of variables from~$\X$ that~$f$ is applied to. A formula consisting of several constraint applications is called an~\emph{$\F$-formula} if all constraint applications come from a constraint in~$\F$.

Though many variants of CSPs exist, the goal in such problems is typically related to the satisfiability of a given formula~$F$. An \emph{assignment}~$s:\X \to \{0,1\}$ over variable set~$\X$ is a function that maps every variable in this set to a value in the domain (in our setting the Boolean domain~$\{0,1\}$). Then, an assignment~$s$ over variable set~$\X$ is said to \emph{satisfy} a constraint application~$(f, \X')$ over~$\X' \subseteq \X$ if~$f(s(\X')) = 1$. We say that an assignment~$s$ satisfies a formula~$F$ if the assignment satisfies all constraint applications in~$F$. Finally, a formula is \emph{satisfiable} if there exists at least one assignment that satisfies it. 

Perhaps the most well-known class of CSPs is the one where, for a given~$\F$-formula~$F$, the task is to determine whether or not~$F$ is satisfiable. Varying the pool of allowed constraints~$\F$ impacts the difficulty of the problem and, in 1978, Schaefer gave a dichotomy theorem that characterizes for which constraint sets~$\F$ the resulting problem is polynomial-time solvable and for which this is NP-hard~\cite{Schaefer78}. Since then, many other variants of CSPs have been studied. We continue this line of research for the following class of optimization problems.

\problem{\wminF}
{An~$\F$-formula~$F$ in which each constraint application has a positive integer weight encoded in unary.}
{Find a minimum-weight set~$X$ of constraint applications in~$F$ s.t.~$F-X$ is satisfiable.}

For an instance~$F$, we denote the weight of an optimal solution by~$\opt(F)$. Additionally, we define the unweighted variant~\minF of the above problem as the restriction in which the weights in the input are all~$1$. Then, the resulting task is equivalent to finding a minimum-\emph{size} set~$X$ for which~$F-X$ is satisfiable.

This problem is NP-hard for many constraint classes~$\F$. As such, e.g.\ approximation algorithms and fixed-parameter tractable (FPT) algorithms parameterized by the solution size~$t$ (i.e. algorithms that decide whether~$\opt(F) \leq t$ in~$g(t) \cdot \mathrm{poly}(n)$ time on formulas of size~$n$, for some computable function~$g$) have been studied. This has resulted in dichotomy theorems for the problem's constant-factor approximability in polynomial time~\cite{KhannaSTW00}, its solvability in FPT time~\cite{KimKPW25}, and its constant-factor approximability in FPT time~\cite{BonnetEM16}.

We study a class of preprocessing algorithms for this problem. Various preprocessing algorithms are known to be very useful in practice as observed in SAT solvers for example~\cite{AchterbergBGRW16}, motivating a study of the extent to which polynomial-time preprocessing can aid in solving \wminF. Preprocessing with the goal of reducing to an instance whose size is polynomially bounded in~$\opt(F)$ has been studied for several problems through the notion of kernelization~\cite{FominLSM19}. These studies include, for example, the \textsc{Min $2$CNF-Deletion} problem (a.k.a. \textsc{Almost $2$-SAT}~\cite{KratschW20}). Here, we take a different perspective on preprocessing with the goal of investigating how preprocessing can help for inputs whose solutions are large, that is, of size comparable to the total input size.

\subparagraph{Essential constraint applications} We extend a recent framework by Bumpus et al.~that introduces the notion of~$c$-essential objects~\cite{BumpusJK24}. Although originally defined within the context of vertex selection problems on graphs, we can rephrase their original definition to apply to~\minF problems. Formally, for a given constant~$c \in \Rone$ and an~$\F$-formula~$F$, we call a constraint application in this formula~\emph{$c$-essential} if it is contained in all~$c$-approximate solutions to the~\textsc{(Weighted)}~\minF instance~$F$. 

One could expect such~$c$-essential objects to somehow stand out in the solution space, since we cannot even form a~$c$-approximation without them. This might make it feasible to identify such objects during preprocessing, which would be useful for problems like \minF in which the goal is to arrive at a certain structure by deleting a minimum number of objects. All essential objects are in particular part of all optimal solutions, so that, after finding and reducing them from the input, it suffices to solve an instance that looks for a strictly smaller solution. Since many FPT algorithms have a running time that scales exponentially with the solution size, each element of the solution we find during preprocessing decreases the running time by a multiplicative factor. As such, we are interested in preprocessing that finds essential objects, if there are any. This algorithmic task is formalized in the following problem statement, defined for a fixed constraint set~$\F$ and constant~$c \in \Rone$.

\problem{$c$-essential Detection for \wminF}
{An~$\F$-formula~$F$ with associated weights, and an integer~$k$}
{Find a subset~$S$ of constraint applications in~$F$ such that:
\begin{enumerate}[(G1)]
    \item \label{itm:G1} if~$\opt(F) \leq k$, then there is an optimal solution to~$F$ containing all of~$S$, and
    \item \label{itm:G2} if~$\opt(F) = k$, then~$S$ contains all~$c$-essential constraint applications.
\end{enumerate}
}

Again we define the unweighted variant of the problem by requiring all weights in the input to be~$1$. Note that this problem is, in two ways, easier than the task of simply returning a set containing exactly the~$c$-essential constraint applications in an input. First, the output set is allowed to contain more constraint applications, as long as they belong to an optimal solution. Secondly, the algorithm only needs to give a meaningful output if the integer~$k$ in its input provides a correct guess or upper bound on the optimal solution of the input. Nonetheless, being able to solve this detection task suffices to speed up FPT algorithms parameterized by solution size, in a black-box fashion.
In particular, if for a constraint set~$\F$ and a constant~$c \in \Rone$, (1) there is a polynomial-time algorithm for~\textsc{$c$-essential Detection for \minF}, and (2) there is an algorithm that outputs a solution to \minF of size at most~$p$, if one exists, in~$g(p) \cdot \mathrm{poly}(n)$ time for some computable function~$g$, then there is an algorithm that computes an optimal solution of \minF in~$g(p-q) \cdot \mathrm{poly}(n)$ time, where~$q$ is the number of~$c$-essential constraint applications in the input and~$p = \opt(F)$. Bumpus et al.\ have proven this guarantee in the setting of graph problems, but their proof is easily seen to be valid for \minF as well \cite[Theorem 5.1]{BumpusJK24}.

Using this framework, several positive results and hardness results have been achieved for graph problems such as \textsc{Vertex Multicut} and \textsc{Directed Feedback Vertex Set}~\cite{BumpusJK24,JansenV26}. Here, we extend the framework to achieve new results for CSPs.

\subparagraph{Our results and other dichotomies} We present a dichotomy theorem that characterizes for which constraint sets~$\F$ there is a constant~$c$ such that~\textsc{$c$-essential Detection for \minF} is polynomial-time solvable. Specifically, we show that this is true for constraint sets~$\F$ that are~$0$-valid,~$1$-valid, IHS-B, or bijunctive (these terms are defined in \cref{sec:preliminaries}). The dichotomy reveals that all other constraint sets~$\F$ yield a \minF problem for which $c$-essential detection is hard for every~$c \in \Rone$, under established hardness assumptions. 

\begin{theorem} \label{thm:dichotomy}
    Let $\F$ be a finite Boolean constraint set.
    \begin{enumerate}
        \item \label{itm:dichotomy-1} If $\F$ is $0$-valid, $1$-valid, or $2$-monotone, then there is a polynomial-time algorithm for \textsc{$c$-essential Detection for \minF} for any ~$c \in \Rone$, since even \wminF is polynomial-time solvable in this case~\cite{KhannaSTW00}.
        \item \label{itm:dichotomy-2} Otherwise, if $\F$ is IHS-B or bijunctive, then there is a constant $c_\F$ such that \textsc{$c_\F$-essential Detection for \wminF} is solvable in polynomial time.
        \item \label{itm:dichotomy-3} Otherwise, if $\F$ is affine, then \textsc{$c$-essential Detection for \minF} is NP-hard for every constant~$c \in \Rone$ under the Unique Games Conjecture.
        \item \label{itm:dichotomy-4} Otherwise, if $\F$ is weakly positive or weakly negative, then \textsc{$c$-essential Detection for \minF} is NP-hard for every constant $c \in \Rone$.
        \item \label{itm:dichotomy-5} Otherwise, \textsc{$c$-essential Detection for \minF} is NP-hard for every constant $c \in \Rone$, even on formulas~$F$ with~$\opt(F) \leq 1$.
    \end{enumerate}
\end{theorem}

As mentioned, other dichotomies exist that characterize the efficient solvability and approximability of \minF. \cref{tab:comparison} shows a comparison between some of these that closely relate to our work. It focuses on the boundary between positive and hardness results. The table includes, in left-to-right order, results on constant-factor approximating the problem in polynomial time~\cite{KhannaSTW00}, exactly solving the problem in FPT time, parameterized by solution size~\cite{KimKPW25}, and constant-factor approximating the problem in FPT time~\cite{BonnetEM16}. A summary of our results from \cref{thm:dichotomy} is included in the final column. In a given row, the table indicates whether the specified tasks for constraint sets from that category --- and not from any category from a previous row --- are tractable under a suitable hardness assumption.

\begin{table}[t]
\centering
\begin{tabular}{p{11em}||m{4.5em}|m{7em}|m{4.5em}||m{4em}}
Restriction on~$\F$ & \makecell{Poly-time\\ approx.~\cite{KhannaSTW00}} & \makecell{FPT \\ exact~\cite{KimKPW25}} & \makecell{FPT\\ approx.~\cite{BonnetEM16}} & \makecell{Essential\\ detection}   \\ \hline\hline
$0$-valid, $1$-valid, \mbox{$2$-monotone} & \makecell{\emph{easy}}& \makecell{\emph{easy}}      & \makecell{\emph{easy}} & \makecell{\emph{easy}} \\ \hline
IHS-B                               & \makecell{\emph{easy}} & \centering{\emph{sometimes easy*}} & \makecell{\emph{easy}} & \makecell{\emph{easy}} \\ \hline
Bijunctive                         & \makecell{hard}        & \centering{\emph{sometimes easy*}} & \makecell{\emph{easy}} & \makecell{\emph{easy}} \\ \hline
None of the above                  & \makecell{hard}        & \makecell{hard}                        & \makecell{hard}  & \makecell{hard}               
\end{tabular}
\caption{A comparison of dichotomies for results on solving or approximating \minF. *Solving the problem on IHS-B or bijunctive constraint sets is FPT if an underlying graph (respectively the arrow graph or Gaifman graph) of the constraint set is~$2K_2$-free. It is~W[$1$]-hard otherwise~\cite{KimKPW25}.} \label{tab:comparison}
\end{table}

We point out that~$c$-essential detection appears to be a strictly easier task than approximating the problem in polynomial-time or exactly solving the problem in FPT time. It is in particular interesting to see that all cases that admit an exact FPT algorithm also admit a polynomial-time~$c$-essential detection algorithm. Based on previous work on~$c$-essential objects mentioned earlier, this implies that all these FPT algorithms for~\minF can be updated to yield an improved runtime dependence on the complexity parameter. We also note that our results show the same boundary of hardness as the existing dichotomy on constant-factor approximating the problem in FPT time.

Finally, we mention some of the existing work on preprocessing, specifically kernelization and compression, for \csp\ and related problems. Unlike kernelization, compression allows problem instances to be transformed to an instance of a possibly different problem. Jansen and Włodarczyk have characterized the best-possible compression size achievable for \minF in terms of the number of variables~\cite{JansenW24}. Related CSPs include determining whether a given~$\F$-formula can be satisfied by an assignment that sets at least, or at most, a given number~$k$ of variables to~$1$. Kratsch and Wahlström characterized for which constraint sets the minimization problem admits a kernel of size polynomial in~$k$~\cite{KratschW10}. Later, such a characterization was also given for the maximization problem~\cite{KratschMW16}.

\subparagraph{Techniques} We proceed by giving a brief overview of the techniques we use to obtain our results. Our algorithmically most interesting positive result is the polynomial-time~$c$-essential detection algorithm on bijunctive constraint sets. To achieve it, we use the insight that bijunctive formulas admit a representation in graph form. Then, for each constraint application in the input formula, we solve a sequence of separation problems on this graph to determine whether to put it in the output set of our~$c$-essential detection algorithm or not.

Next, for each of our hardness results, we provide a two-step approach. To show, for a given category of constraint sets, that~$c$-essential detection is hard for \minF, we start by formulating a canonical constraint set from that category and showing that the problem is hard for that set. Then, we show that the corresponding problem for all other constraint sets in that category is as hard as the problem for the canonical constraint set.

\newcommand{\Fcan}{\F_{\mathrm{can}}}
To achieve this second step, we show that an \emph{essential implementation} of the canonical set~$\Fcan$ for a category can be made via any constraint set~$\F$ in that category. That is, for every constraint application over~$\Fcan$, there is a conjunction of constraint applications over~$\F$ that is equivalent to it, while satisfying an additional property as defined in \cref{def:essential-implementation}. The latter ensures that for every constant~$c$, a reduction can be made from~\mincsp{$\Fcan$}to~\wminF that preserves the entire space of~$c$-approximate solutions. 

We remark that our definition of essential implementations is strictly stronger than the very similar notion of strong and perfect implementations by Khanna et al.~\cite{KhannaSTW00}. An implementation from their framework also yields a reduction between CSP instances over different constraint sets. They showed that this reduction preserves the constant-factor approximability of the involved problems, but, in contrast to our result, it need not preserve the exact approximation ratio. We were surprised to see that most of the implementations we encountered in related work (also beyond strong and perfect implementations) happen to fit our definition of an essential implementation as well. This reveals even stronger links between problems that have already been shown to be reducible to one another.

\subparagraph{Organization} The rest of the paper is organized as follows. In \cref{sec:preliminaries}, we define some fundamental concepts and notation. \cref{sec:positive} contains our positive results, proving Case~\ref{itm:dichotomy-2} of \cref{thm:dichotomy}. \cref{sec:hardness} contains our hardness results, starting with an introduction to essential implementations. Then, in \cref{ssec:affine,ssec:wpos-wneg,ssec:other}, we prove Cases~\ref{itm:dichotomy-3},~\ref{itm:dichotomy-4}, and~\ref{itm:dichotomy-5} of \cref{thm:dichotomy} respectively. We conclude with open questions in \cref{sec:conclusion}.

\section{Preliminaries}
\label{sec:preliminaries}
A constraint is a Boolean function~$f \colon \{0,1\}^p \to \{0,1\}$ for some integer~$p$, which is its \emph{arity}. 
We assume all individual Boolean constraints to be satisfiable: unsatisfiable constraint applications in a formula~$F$ must be part of every solution to~$F$ and could thus be filtered in a simple preprocessing step. We call two constraints~$f_1$ and~$f_2$ \emph{equivalent} if they have the 22same arity and are satisfied by the same set of assignments. We denote this as~$f_1 \equiv f_2$.

Next, we define some standard Boolean constraints that are used throughout the paper. We write~$\true$ or~$\false$ for the arity-$1$ constraint that is satisfied if and only if its input variable is respectively~$1$ or~$0$. For non-negative integers~$p$ and~$q \leq p$, we denote the \mbox{arity-$p$} constraint that takes the logical OR over $p$ literals, of which $q$ are negated, by~$\orr_{p,q}$. Hence,~$\orr_{p,q}$ denotes the constraint~$(\neg x_1 \vee \ldots \vee \neg x_q \vee x_{q+1} \vee \ldots \vee x_p)$. We use shorthands~$\orr_p$ for~$\orr_{p,0}$, and~$\nor_p$ for~$\orr_{p,p}$. We denote the constraint~$(x_1 \oplus \ldots \oplus x_p = 1)$ by~$\xor_p$ and the constraint~$(x_1 \oplus \ldots \oplus x_p = 0)$ by~$\xnor_p$. We use shorthands~$\neqq$ for~$\xor_2$ and~$\eqq$ for~$\xnor_2$. We denote the constraint of arity~$p$ that is satisfied if and only if its~$p$ input variables are not all equal by~$\nae_p$.

Now, we can define different classes of constraints. First, we call a constraint~\emph{$0$-valid} or~\emph{$1$-valid} if it is satisfied by the all-$0$ or the all-$1$ assignment respectively. We call a constraint \emph{IHS-B$^+$} if it can be expressed as a conjunction of~$\orr_p$ constraints for integers~$p$, the constraint~$\orr_{2,1}$, and the arity-$1$ constraint~$\false$. Likewise, we call constraints \emph{IHS-B$^-$} if they can be expressed as conjunction of~$\nor_p$ constraints for integers~$p$, the constraint~$\orr_{2,1}$, and the arity-$1$ constraint~$\true$. We call a constraint \emph{bijunctive} if it can be expressed in $2$CNF (i.e.: as conjunction of~$\orr_2$,~$\orr_{2,1}$, and~$\nor_2$ constraints). As subset of the bijunctive constraints, $2$-monotone constraints are defined as those that can be expressed in the form~$(a_1 \wedge \ldots \wedge a_p) \vee(\neg b_1 \wedge \ldots \wedge \neg b_q)$ for some~$p,q \geq 0$. We call a constraint \emph{affine} if it can be expressed as conjunction of linear constraints mod~$2$ (i.e.: as conjunction of~$\xor_p$ and~$\xnor_q$ constraints for integers~$p\geq1$ and integers~$q\geq1$). We call a constraint \emph{weakly positive} (resp.\ \emph{weakly negative}) if it can be expressed in CNF with all clauses containing at most one negated (resp.\ positive) variable.

We call a constraint set~$\F$~$0$-valid / $1$-valid / bijunctive / $2$-monotone / affine / weakly positive / weakly negative if every constraint~$f \in \F$ satisfies the corresponding property. We call~$\F$ IHS-B if all its constraints are IHS-B$^+$ or all its constraints are IHS-B$^-$ and we call a single constraint~$f$ IHS-B if it is either IHS-B$^+$ or IHS-B$^-$.

For an assignment~$s$ to~$\X$, we define its \emph{complement} as the assignment~$\overline{s}$ to~$\X$ that satisfies~$\overline{s}(x) = 1-s(x)$. If the complement of every satisfying assignment for a constraint (application)~$f$ also satisfies~$f$, we call it a \emph{complementive} constraint (application). If all constraints in a constraint set~$\F$ are complementive, we call~$\F$ complementive as well. 

We denote the all-zero (resp. all-one) assignment to~$\X$ by~$K_0$ (resp.~$K_1$). For a subset~$V \subseteq \X$, we denote the assignment that assigns~$1$ (resp.~$0$) to all variables in~$V$ and~$0$ (resp.~$1$) to all variables in~$\X \setminus V$ by~$K_{0,V}$ (resp.~$K_{1,V}$).

For a constraint~$f$ over~$\X$, pairwise disjoint subsets~$V_1, \ldots, V_p$, and variables~$v_1, \ldots, v_p$ (that are not necessarily in~$\X$), we use~$f\left[ \stack{V_1}{v_1} \comma \ldots \comma \stack{V_p}{v_p} \right]$ to denote the constraint obtained by replacing every occurrence of a variable from some~$V_i$ by~$v_i$. If a set~$V_i$ is a singleton set, we may replace it by its singleton element in this notation instead. We also allow a variable~$v_i$ to be either of the constants~$0$ or~$1$ in this notation, and we use this to denote a constraint obtained from~$f$ by replacing some of this variables by the respective constant.

\subsection{Weighted CSPs}
Next, we note that we can use a reduction that transforms integer-weighted formulas into unweighted ones by duplicating constraint applications as many times as their weight, to show the following.

\begin{lemma}
\label{lem:unweighted-to-weighted}
    Let~$\F$ be a constraint set and~$c \in \Rone$. If \minF admits a polynomial-time $c$-essential detection algorithm, then so does \wminF if all weights are integers that are polynomial in the number of constraint applications.
\end{lemma}
\begin{proof}
    We prove the statement by providing a reduction from the weighted to the unweighted setting. Since all weights in the weighted setting are integers, we replace every constraint application~$C$ with weight~$w(C)$ in the weighted formula $F$ by~$w(C)$ unweighted copies~$C_1, \ldots, C_{w(C)}$ of that same constraint application. This yields an equivalent unweighted formula~$F'$. First, we prove the following two claims, which show that the solution space is, in some way, preserved.

    \begin{claim} \label{clm:opt-uses-all-copies}
        Let~$S'$ be an inclusion-wise minimal solution for~$F'$, and $C$ a constraint application in $F$. Then,~$S'$ either contains all or none of the copies~$C_1, \ldots, C_{w(C)}$.
    \end{claim}
    \begin{claimproof}
        Suppose, for contradiction, that~$S'$ contains some but not all copies of~$C$. Then, since~$S'$ is a solution, there is an assignment~$s$ that satisfies all constraint applications in~$F'-S'$. This includes at least one copy of $C$, since $S'$ does not contain all of them.~$S'$ does contain at least one copy of~$C$, which is clearly also satisfied by~$s$, so removing this copy from~$S'$ yields a strictly smaller solution. This contradicts the minimality of~$S'$.
    \end{claimproof}

    \begin{claim} \label{clm:weight-reduction-preserves-solutions}
        A set~$S$ is an (optimal) solution for~$F$ if and only if the set~$S' := \bigcup_{C \in S} \{C_1, \ldots, C_{w(C)}\}$ is an (optimal) solution for~$F'$. Hence, $\opt(F) = \opt(F')$.
    \end{claim}
    \begin{claimproof}    
        It is evident that~$F-S$ is satisfiable if and only if~$F'-S'$ is satisfiable: both formulas contain the exact same set of constraint applications.~$F'-S'$ might contain more copies of the same constraint application, but the two formulas are equisatisfiable nonetheless. 
        
        Moreover, by \cref{clm:opt-uses-all-copies}, if a solution to~$F'$ is optimal, it is always of this specific form: consisting, for a given set of constraint applications in~$F$, precisely of their copies in~$F'$. Since the number of copies of each constraint application in~$F'$ equals the weight of that constraint application in~$F$, the total weight of~$S$ equals the size of~$S'$. Hence, $\opt(F) = \opt(F')$.
    \end{claimproof}

    Now, to prove the correctness of our reduction, we use it to give a polynomial-time $c$-essential detection algorithm for the original weighted formula. Given a weighted formula~$F$ and an integer~$k$, we first apply the reduction to turn it into an unweighted formula~$F'$ as above. Then, let~$D'$ be the result of a polynomial-time $c$-essential detection algorithm on~$F'$ with the given integer~$k$. By assumption, such an algorithm exists. Then, we construct the set~$D$ as the set containing all constraint applications~$C$ in~$F$ for which at least one of its copies~$C_1, \ldots, C_{w(f)}$ is in~$D'$. The remainder of this proof is used to show that outputting the resulting set~$D$ yields a correct $c$-essential detection algorithm. This proves the lemma as the proposed procedure can clearly be executed in polynomial -time.

    To show that~$D$ satisfies Property~\textbf{(G\ref{itm:G1})}, suppose that $k \leq \opt(F) = \opt(F')$. Then, since~$D'$ satisfies Property~\textbf{(G\ref{itm:G1})}, there is an optimal solution~$S'$ of~$F'$ that is a superset of~$D'$. By \cref{clm:opt-uses-all-copies},~$S'$ contains none or all of the copies of each constraint application. Then, by \cref{clm:weight-reduction-preserves-solutions}, the set~$S$ of constraint applications in~$F$ for which~$D'$ contains its copies in~$F'$ is an optimal solution for~$F$. Since~$S'$ is a superset of~$D'$, it follows by construction of~$D$ that~$S$ is a superset of~$D$. Hence,~$D$ satisfies Property~\textbf{(G\ref{itm:G1})}.

    To show that~$D$ satisfies Property~\textbf{(G\ref{itm:G2})}, suppose that~$k = \opt(F) = \opt(F')$. Suppose, for contradiction, that there is some $c$-essential constraint application~$C$ that is not in~$D$. Then, note that there is also at least one copy~$C_i$ of it that is not in~$D'$. Since~$D'$ satisfies Property~\textbf{(G\ref{itm:G2})}, it contains all $c$-essential constraint applications of $F'$, meaning that~$C_i$ is not $c$-essential. As such, there is a $c$-approximate solution for~$F'$ that does not contain~$C_i$. Let~$A'$ be such a solution that also satisfies the property of containing none or all of the copies of each constraint application in~$F$. It follows from \cref{clm:opt-uses-all-copies}, that such a solution exists. Now, let~$A$ be the set of constraint applications in~$F$ of which their copies in~$F'$ are contained in~$A'$. By \cref{clm:weight-reduction-preserves-solutions},~$A$ is a solution for~$F$. Note that~$A$ does not contain~$C$. Moreover, since~$A'$ (whose size equals the weight of $A$) is a $c$-approximate solution for~$F'$ and $\opt(F) = \opt(F')$, we find that~$A$ is a $c$-approximate solution for~$F$. Since it does not contain~$C$, this contradicts the assumption that $C$ is $c$-essential. As such,~$D$ contains all $c$-essential constraint applications, thereby satisfying Property~\textbf{(G\ref{itm:G2})}.
\end{proof}

\subsection{Characterizations of constraint classes} \label{ssec:prelim-characterizations}
To prove Case~\ref{itm:dichotomy-5} of \cref{thm:dichotomy} in \cref{ssec:other}, we use several known characterizations of constraint classes. These are given below in \cref{lem:wpos-wneg-characterization,lem:affine-characterization,lem:bijunctive-characterization}. Before stating them, we introduce the required notation.

Let~$f$ be a constraint over~$\X$ and let~$V \subseteq \X$. We say that~$V$ is~\emph{$0$-consistent} (resp.~$1$-consistent) if there is an assignment that satisfies~$f$ and sets all variables in~$V$ --- and possibly some others --- to~$0$ (resp.~$1$). We say that~$V$ is~\emph{$0$-closed} (resp.~\emph{$1$-closed}) if, for every variable~$x \in \X \setminus V$, there is an assignment that satisfies~$f$ and sets all variables in~$V$ to~$0$ (resp.~$1$) and~$x$ to~$1$ (resp.~$0$). Alternatively,~$V$ is also called~$0$-closed (resp.~$1$-closed) if no satisfying assignment to~$f$ sets all variables in~$V$ to~$0$ (resp.~$1$). However, in this work, we happen to only refer to~$0$-closed (resp.~$1$-closed) sets that are also~$0$-consistent (resp.~$1$-consistent). If~$s$ is a satisfying assignment for~$f$, we say that~$V$ is a \emph{change set} for~$(f,s)$ if~$s \oplus K_{1,V}$ (i.e. the assignment to~$\X$ that only differs from~$s$ in the variables~$V$) also satisfies~$f$.

Using this notation, we can give three characterizations of different constraint classes.

\begin{lemma}[{\cite[Lemma 3.1W]{Schaefer78}}] \label{lem:wpos-wneg-characterization}
    Let~$f$ be a Boolean constraint over variables~$\X$. Then,~$f$ is weakly positive if and only if for all subsets~$V \subseteq \X$ that are both $0$-closed and $0$-consistent the assignment~$K_{0,V}$ satisfies~$f$. Likewise,~$f$ is weakly negative if and only if for all subsets~$V \subseteq \X$ that are both $1$-closed and $1$-consistent the assignment~$K_{1,V}$ satisfies~$f$.
\end{lemma}

\begin{lemma}[{\cite[Lemma 3.1A]{Schaefer78}}] \label{lem:affine-characterization}
    A Boolean constraint~$f$ is affine if and only if for all satisfying assignments~$s_1, s_2, s_3$ the assignment~$s_1 \oplus s_2 \oplus s_3$ also satisfies~$f$.
\end{lemma}

\begin{lemma}[{\cite[Lemma 3.1B]{Schaefer78}}] \label{lem:bijunctive-characterization}
    A Boolean constraint~$f$ is bijunctive if and only if for all satisfying assignments~$s$ and change sets~$U, V$ for~$(f,s)$ the intersection~$U \cap V$ is also a change set for~$(f,s)$.
\end{lemma}

\section{Positive results}
\label{sec:positive}
To achieve our positive results, we establish conditions on a constraint set~$\F$ that imply the existence of a~$c$-essential detection algorithm for \minF. The statement below admits a proof similar to that of Theorem 4.1 in the earlier work by Bumpus et al.\ on essential vertices~\cite{BumpusJK24}.
\begin{lemma}
\label{lem:main-positive-ingredient}
    Let~$\F$ be a constraint set and let~$c \in \Rone$ be a constant. Suppose there is a polynomial-time algorithm that takes as input an $\F$-formula~$F$ and a constraint application~$C$ in it, and outputs a set of constraint applications~$P_C$ such that:
    \begin{enumerate}
        \item \label{itm:positive-prerequisite-1} $C \notin P_C$; and 
        \item \label{itm:positive-prerequisite-2} if~$X_C$ is a smallest set of constraints from~$F-C$ such that~$F-X_C$ is satisfiable, then $|P_C| \leq c \cdot|X_C|$; and
        \item \label{itm:positive-prerequisite-3} for every set~$X$ such that~$F-X$ is satisfiable,~$F-((X \setminus \{C\}) \cup P_C)$ is also satisfiable,
    \end{enumerate}
    then, there is a polynomial-time $(c+1)$-essential detection algorithm for \minF.
\end{lemma}
\begin{proof}
    Let~$F$ be an $\F$-formula, let~$k \geq 0$ be an integer, and consider the following algorithm to compute a set~$S$ of constraint applications. We iterate over all constraint applications~$C$ in~$F$ to compute, in polynomial time, a set~$P_C$ satisfying the three properties specified in the lemma statement. Starting with an empty set~$S$, we add all constraint applications~$C$ to it for which $|P_C| > c\cdot k$. Finally, we return the set~$S$.

    Clearly, this is a polynomial-time algorithm. To show that it is a~$(c+1)$-essential detection algorithm for \minF, we proceed by proving that~$S$ satisfies Properties~\textbf{(G\ref{itm:G1})} and~\textbf{(G\ref{itm:G1})}. To this end, let~$X$ be an optimal solution to~$F$.

    For Property~\textbf{(G\ref{itm:G1})}, suppose that~$\opt(F) \leq k$, so that~$|X| \leq k$. We show that~$S \subseteq X$ by arguing that every~$C\notin X$ is not in~$S$. Let~$C \notin X$ be arbitrary. Since~$X$ is a smallest set of constraint applications such that~$F-X$ is satisfiable, it is in particular a smallest such set that does not include~$C$. Therefore, by property~\ref{itm:positive-prerequisite-2} of~$P_C$, we find that~$|P_C| \leq c \cdot |X| \leq c\cdot k$. As we only added constraints~$C$ to~$S$ for which~$|P_C| > c\cdot k$,~$C$ is not in~$S$.

    For Property~\textbf{(G\ref{itm:G2})}, suppose that~$\opt(F) = k$, so that~$|X| = k$. We show that~$S$ contains all~$(c+1)$ essential constraint applications by showing that, if a constraint application~$C$ is not in~$S$, it is not~$(c+1)$-essential, i.e.: there is a~$(c+1)$-approximate solution to~$F$ that does not include~$C$. If~$C \notin X$, then~$X$ is such a solution. If~$C \in X$, we argue that~$(X \setminus \{C\}) \cup P_f$ is such a solution. By property~\ref{itm:positive-prerequisite-1} of~$P_C$, this set does not contain~$C$ and by property~\ref{itm:positive-prerequisite-3}, it is indeed a solution. To see that it is a~$(c+1)$-approximate solution, recall that~$C \notin S$, which means that~$|P_C| \leq c \cdot|X| = c\cdot k$. Therefore~$|(X \setminus C) \cup P_C| \leq k+c\cdot k = (c+1)\cdot k$.
\end{proof}

Now, we can see the following observation to be true by noting that its precondition satisfies all three preconditions of \cref{lem:main-positive-ingredient}.
\begin{observation} \label{obs:f-avoiding-approx-means-essential-detection}
    Let~$\F$ be a constraint set and let~$c \in \Rone$ be a constant. Suppose there is a polynomial-time algorithm that takes as input an~$\F$-formula~$F$ and a constraint application~$C$ in~$F$ and outputs a $c$-approximate solution to the problem of finding a smallest solution for~$F$ that does not contain~$C$. Then, there is a polynomial-time~$(c+1)$-essential detection algorithm for \minF.
\end{observation}

In turn, we can use the observation above to prove the following statement.

\begin{lemma} \label{lem:weighted-approx-means-essential-detection}
    Let~$\F$ be a constraint set and let~$c \in \Rone$ be a constant. If there is a polynomial-time~$c$-approximation for \wminF, then there is a polynomial-time~$(c+1)$-essential detection algorithm for \minF.
\end{lemma}
\begin{proof}
    We show that this condition suffices to satisfy the precondition to \cref{obs:f-avoiding-approx-means-essential-detection}.
    Given an~$\F$-formula~$F$ and a constraint application~$C$ in~$F$, we transform~$F$ into a weighted formula by giving all~$n$ constraint applications in~$F$ a weight of~$1$, except for~$C$, which receives a weight of~$c\cdot n$. Then, the polynomial-time~$c$-approximation for \wminF on this weighted formula returns a~$c$-approximate solution to the problem of finding a smallest solution for the original unweighted formula~$F$ that does not include~$f$. Thus, by \cref{obs:f-avoiding-approx-means-essential-detection}, there is a polynomial-time~$(c+1)$-essential detection algorithm for \minF.
\end{proof}

Khanna et al.\ have shown that a polynomial-time constant-factor approximation algorithm exists when~$\F$ is IHS-B \cite[Theorem 2.13]{KhannaSTW00}, so we conclude the following.

\begin{corollary}
    Let~$\F$ be an IHS-B constraint set. Then, there is a constant~$c_\F$ for which \minF admits a polynomial-time~$c_\F$-essential detection algorithm.
\end{corollary}

The remainder of this section is dedicated to proving a similar result for bijunctive constraint sets~$\F$. This proof requires more work but will also use \cref{lem:main-positive-ingredient}. In addition, we make use of the well-known fact~\cite{AspvallPT79,RamanujanS17} that formulas in $2$CNF can be represented in graph form as follows.
\begin{definition}
    Let~$F$ be a $2$CNF formula over variable set~$\X$, or more generally, an~$\mathcal{F}$-formula for some bijunctive constraint set~$\mathcal{F}$. We construct a directed (multi)graph~$G_F$ as follows. For every variable~$x \in \X$, we add~$x$ and~$\neg x$ to the vertex set of~$G_F$. For every clause~$(\ell_1 \vee \ell_2)$ in the $2$CNF representation of~$F$ with literals~$\ell_1$ and $\ell_2$, we add the arcs~$(\neg \ell_1, \ell_2)$ and~$(\neg \ell_2, \ell_1)$ to~$G_F$. We call the resulting graph the \emph{implication graph} of~$F$.
\end{definition}
We highlight that clauses may appear in multiple constraint applications of~$F$ and that this is reflected in the implication graph of~$F$ by allowing duplicate arcs in it. Furthermore, in our remaining arguments, we assume that we keep track of which clause in a formula corresponds to which arc in the implication graph. We proceed by presenting two basic properties regarding implication graphs in the following two lemmas.
\begin{lemma}
\label{lem:self-reachability-requires-endpoint}
    Let~$F$ be a $2$CNF formula over variable set~$\X$ and let~$q$ be a literal of some variable in~$\X$. If~$\neg q$ is reachable from~$q$ in the implication graph~$G_F$ of~$F$, then any assignment that satisfies~$F$ must set~$q$ to~$0$.
\end{lemma}
\begin{proof}
    Suppose, for contradiction, that~$s$ is a satisfying assignment with~$s(q) = 1$ and let~$P$ be an~$(q, \neg q)$-path in~$G_F$. We prove, using induction, that all literals on this path are assigned~$1$ by~$s$. To this end, consider some literal~$\ell_i$ on~$P$ with~$s(\ell_i) = 1$ and let~$\ell_{i+1}$ be its successor on~$P$. The existence of the arc~$(\ell_i, \ell_{i+1})$ in~$G_F$ implies the existence of the clause~$(\neg \ell_i \vee \ell_{i+1})$ in~$F$. Since~$s$ satisfies~$F$ with~$s(\ell_i) = 1$, it must have~$s(\ell_{i+1}) = 1$. Combined with the fact that the start vertex~$q$ of~$P$ has~$s(q)=1$, we can use this argument inductively to conclude that all vertices on~$P$ --- including its ending vertex~$\neg q$ --- are assigned~$1$ by~$s$. Then, the fact that~$s(\neg q) = 1$ contradicts our assumption that~$s(q) = 1$.
\end{proof}
\begin{lemma}
\label{lem:reverse-paths-implication-graph}
    Let~$G_F$ be the implication graph of a $2$CNF formula~$F$ over variables~$\X$, and let~$q$ and~$r$ each be a literal of a variable in~$\X$. If~$r$ is reachable from~$q$ in~$G_F$, then $\neg q$ is reachable from~$\neg r$ in~$G_F$.
\end{lemma}
\begin{proof}
    Let~$(x,y)$ be an arbitrary arc in~$G_F$. The existence of this arc in~$G_F$ implies that~$F$ contains the clause~$(\neg x \vee y)$, which means that~$G_F$ also contains the arc~$(\neg y, \neg x)$. Since for every arc~$(x,y)$ in~$G_F$ the graph also contains the arc~$(\neg y, \neg x)$, the existence of some~$(q,r)$-path~$(q, \ell_1, \ldots, \ell_p,r)$ in~$G_F$ implies the existence of the~$(\neg r, \neg q)$-path~$(\neg r, \neg \ell_p, \ldots, \neg \ell_1, \neg q)$ in~$G_F$.
\end{proof}
Now, we use these properties to obtain the following result.
\begin{lemma} \label{lem:bijunctive}
    Suppose that~$\F$ is bijunctive so that every constraint~$f \in \F$ can be written as conjunction of at most~$d$ disjunctions of two literals, for some constant~$d$. Then, there is a polynomial-time~$(2d^2+1)$-essential detection algorithm for \minF.
\end{lemma}
\begin{proof}
    First note that we may assume w.l.o.g.\ that input formulas~$F$ are already given to us in suitable $2$CNF representation when dealing with \minF instances. Otherwise, we could, for every fixed bijunctive~$\F$, provide a polynomial-time transformation using a hard-coded bijunctive representation of each constraint in~$\F$. Thus, let~$F$ be an $\F$-formula represented in $2$CNF in which every constraint is written as a conjunction of at most~$d$ clauses. Let~$G_F$ be its implication graph, and let~$C$ be a constraint application in~$F$ that is expressed as $(a_1 \vee b_1) \wedge \ldots \wedge(a_\ell \vee b_\ell)$, where~$a_1, \ldots a_\ell, b_1, \ldots, b_\ell$ are literals for some~$\ell \leq d$.

    To prove the lemma, we give a polynomial-time algorithm that computes a set $P_C$ from~$F$ and~$C$, after which we show that this set meets the three preconditions described in \cref{lem:main-positive-ingredient}. To ensure that our algorithm is well-defined, we prove the claim below.
    
    \begin{claim}
    \label{clm:path-absence}
        Let~$F'$ be any 2CNF formula that contains~$C$ and let~$G_{F'}$ be the implication graph of~$F'$. If there is an~$i \in [\ell]$ such that~$G_{F'}$ contains both an~$(a_i, \neg a_i)$-path and a~$(b_i, \neg b_i)$-path, then~$F'$ is unsatisfiable.
    \end{claim}
    \begin{claimproof}
        Let~$i \in [\ell]$ be such that~$G_{F'}$ contains both an~$(a_i, \neg a_i)$-path and a~$(b_i, \neg b_i)$-path. Then, by \cref{lem:self-reachability-requires-endpoint}, any satisfying assignment for~$F'$ must set~$a_i$ and~$b_i$ to $0$. However, since~$F'$ contains~$C$, which in turn is expressed as a 2CNF constraint application containing the clause~$(a_i \vee b_i)$, every assignment that satisfies~$F'$ must set at least one of~$a_i$ and~$b_i$ to $1$. Hence,~$F'$ is unsatisfiable.
    \end{claimproof}
    Now,~$P_C$ can be computed as follows.
    \begin{itemize}
        \item Construct the implication graph~$G_{F}$ of~$F$.
        \item For all~$i \in [\ell]$, compute a smallest subset of arcs in $G_{F-f}$ that breaks all~$(a_i, \neg a_i)$-paths in~$G_F$ and compute a smallest subset of arcs in $G_{F-f}$ that breaks all~$(b_i, \neg b_i)$-paths in~$G_F$. If only one of these two sets exists, let~$T_i$ be this set, or, if both exist, let~$T_i$ be a minimum-size set among the two. (Note that at least one of these cuts must exist. If not, the implication graph~$G_C$ of the singular constraint application~$C$ would have to contain both an~$(a_i, \neg a_i)$-path and a~$(b_i, \neg b_i)$-path, which, by \cref{clm:path-absence}, would imply~$C$ to be unsatisfiable. This would contradict our assumption that all individual constraints we consider are satisfiable.)
        \item Let~$T$ be $\bigcup_{i \in [\ell]} T_i$ and let~$P_C$ be the set of constraint applications in~$F$ with at least one of their corresponding implication arcs in~$T$.
    \end{itemize}
    First, we observe that the construction above can be executed in polynomial time. The bottleneck is the second step, in which a linear number of cut problems need to be solved. These are variants of the standard min-cut problem on directed graphs, which is solvable in polynomial time using the Ford-Fulkerson algorithm for example~\cite{FordF56}. To model that a given set of arcs in the graph (those that correspond to~$C$) is undeletable, we can contract these edges before running the min-cut algorithm. If the two vertices to be cut are contracted to the same vertex, we conclude that no cut in~$G_F$ exists that avoids all of the undeletable arcs.
    
    It remains to prove that $P_C$ satisfies the three properties specified in \cref{lem:main-positive-ingredient}. By construction,~$P_C$ does not contain~$C$, so it satisfies the first property.

    For the other two properties, we start by noting that~$T$ is a set of arcs that, for every~$i \in [\ell]$, breaks either all~$(a_i, \neg a_i)$-paths or all~$(b_i, \neg b_i)$-paths in~$G_{F}$ without including any arcs that correspond to a clause of~$C$. We call such a set a \emph{$C$-respecting arc set}. (Note that such sets may include duplicates of arcs that correspond to a clause of~$C$, but not those that specifically correspond to~$C$.) Likewise, we call a set of constraint applications from~$F$ a \emph{$C$-respecting constraint application set} if the union of their respective arcs in~$G_F$ is a $C$-respecting arc set. As such,~$P_C$ is a $C$-respecting constraint application set. We use this notation to show that~$P_C$ satisfies the second property specified in \cref{lem:main-positive-ingredient}.
    \begin{claim}
        Let~$X_C$ be a smallest set of constraints from~$F-C$ such that~$F-X_C$ is satisfiable. Then,~$|P_C| \leq 2d^2\cdot|X_C|$.
    \end{claim}
    \begin{claimproof}
        First, we argue that~$X_C$ is a $C$-respecting constraint application set. Suppose, for contradiction, that it is not. Then, since~$C \notin X_C$, there is some~$i \in [\ell]$ for which the implication graph~$G_{F-X_C}$ of~$F-X_C$ contains both an~$(a_i, \neg a_i)$-path and a~$(b_i, \neg b_i)$-path. But then, \cref{clm:path-absence} implies that~$F-X_C$ is unsatisfiable; a contradiction.
        
        By definition, the union~$T^\ast$ of all arcs in~$G_F$ that correspond to a clause in~$X_C$ is a $C$-respecting arc set in~$G_F$. Moreover, since every constraint application in~$F$ is expressed as a conjunction of at most~$d$ clauses that each correspond to~$2$ arcs in~$G_F$, we find that~$|T^\ast| \leq 2d \cdot|X_C|$.

        For every~$i \in [\ell]$,~$T^\ast$ breaks either all~$(a_i, \neg a_i)$-paths or all~$(b_i, \neg b_i)$-paths, so $\max_{i \in [\ell]} |T_i| \leq |T^\ast|$. As~$T$ is defined to be the union of the sets~$T_i$ for all $i \in [\ell]$, we get that $|T| \leq \ell \cdot \max_{i \in [\ell]} |T_i|$. Combining this with the previous inequality and the fact that~$\ell \leq d$, we obtain that $|T| \leq d \cdot|T^\ast|$.

        Finally, because every arc in~$G_F$ corresponds to exactly one constraint in~$F$, we find that~$|P_C| \leq |T|$. Combining this with all previous inequalities, we obtain that $|P_C| \leq |T| \leq d \cdot |T^\ast| \leq 2d^2 \cdot |X_f|$.
    \end{claimproof}

    Now, we show that~$P_C$ also satisfies the third property listed in \cref{lem:main-positive-ingredient}.
    \begin{claim}
        For every set~$X$ such that~$F-X$ is satisfiable,~$F^\ast := F-((X \setminus \{C\}) \cup P_C)$ is also satisfiable.
    \end{claim}
    \begin{claimproof}
        Let~$X$ be such that~$F-X$ is satisfiable and let~$s$ be a satisfying assignment for~$F-X$. Of course, this assignment satisfies almost all constraint applications in~$F^\ast$, but it may not satisfy~$C$. First, we present a procedure that transforms~$s$ into a satisfying assignment for~$F^\ast$. Then, we show that this is correct. The procedure is as follows.

        First, select an arbitrary clause~$(a_i \vee b_i)$ from~$C$ that is not yet satisfied by~$s$. Then, determine whether the implication graph~$G_{F^\ast}$ of~$F^\ast$ contains an~$(a_i, \neg a_i)$-path or a~$(b_i, \neg b_i)$-path. Since~$P_C$ is a $C$-respecting constraint set, at least one of these does not exist in~$G_{F^\ast}$. Assume w.l.o.g.\ that~$G_{F^\ast}$ does not contain an~$(a_i, \neg a_i)$-path. Then, update~$s$ by setting~$a_i$ and all literals reachable from~$a_i$ in~$G_{F^\ast}$ to~$1$. Repeat this for as long as there are clauses of~$C$ not yet satisfied by~$s$.

        To show that this procedure is correct, we start by showing that there is no iteration in which we set both~$x$ and~$\neg x$ to~$1$ for any variable~$x$. 
        Consider an arbitrary iteration and suppose, for contradiction, that we set both~$x$ and~$\neg x$ to~$1$ in this iteration. Let~$(a_i \vee b_i)$ be the unsatisfied clause from~$C$ that is selected in this iteration. By assuming w.l.o.g.\ that~$G_{F^\ast}$ does not contain an~$(a_i, \neg a_i)$-path, this would mean that both~$x$ and~$\neg x$ are reachable from~$a_i$ in~$G_{F^\ast}$. However, the fact that~$x$ is reachable from~$a_i$ implies, by \cref{lem:reverse-paths-implication-graph}, that~$\neg a_i$ is reachable from~$\neg x$. Since $\neg x$ is reachable from~$a_i$, this means that~$\neg a_i$ is reachable from~$a_i$. This contradicts the assumption that~$G_{F^\ast}$ does not contain an~$(a_i, \neg a_i)$-path.

        The argument above shows that each iteration ends with a valid assignment~$s$. Also, by setting~$a_i$ to~$1$,~$s$ now satisfies the previously unsatisfied clause~$(a_i \vee b_i)$. Moreover, we show that no clause in~$F^\ast$ that was satisfied at the start of an iteration becomes unsatisfied at the end of the iteration. 
        
        Suppose for contradiction that such a clause~$(x \vee y)$ exists. If this clause was satisfied at the start of the iteration, at least one of~$x$ and~$y$ must have been set to~$0$ during this iteration. Assume w.l.o.g.\ that this was~$x$. Then,~$\neg x$ is reachable from~$a_i$ in~$G_{F^\ast}$. However, the existence of the clause~$(x \vee y)$ in~$f$ implies the existence of an arc~$(\neg x, y)$ in~$G_{F^\ast}$, meaning that~$y$ is also reachable from~$a_i$ in~$G_{F^\ast}$, thus being set to~$1$ during this iteration. This contradicts the assumption that~$(x \vee y)$ is not satisfied at the end of the iteration.

        From this, we conclude that the described procedure to modify~$s$ results in a satisfying assignment of~$F^\ast$: each iteration results in a valid assignment that satisfies strictly more clauses than before. Thus, at some point the procedure terminates, at which point a satisfying assignment of~$F^\ast$ has been constructed. Therefore,~$F^\ast$ is satisfiable.
    \end{claimproof}    

    This proves that the algorithm described above computes a set~$P_C$ in polynomial time that satisfies all three preconditions for \cref{lem:main-positive-ingredient}. In particular, our algorithm satisfies these constraints with~$c=2d^2$. By \cref{lem:main-positive-ingredient}, there is a polynomial-time~$(2d^2+1)$-detection algorithm for \minF.
\end{proof}

We recall that FPT algorithms, parameterized by solution size, are known for \minF on bijunctive constraint sets~$\F$ if and only if the so-called Gaifman graph of~$\F$ is~$2K_2$-free~\cite{KimKPW25}. So, by a result analogous to one from the original work~\cite[Theorem 5.1]{BumpusJK24} on~$c$-essential detection, \cref{lem:bijunctive} effectively shows that we can reduce the search-space of these FPT algorithms. In particular, it shows that their exponential runtime dependence on the total size of an optimal solution can be improved to depend only on the number of constraint applications in an optimal solution that are \emph{not}~$c$-essential.

\section{Hardness results} \label{sec:hardness}
Our hardness results all follow the same overarching structure. For a given category of constraint sets, we first give a canonical set of constraints~$\Fcan$ from that category and prove for the corresponding \mincsp{$\Fcan$} problem that~$c$-essential detection is hard for all constants~$c \in \Rone$. Then, we show that there is a reduction from this problem to every other \minF problem for constraint sets~$\F$ in the same category. This reduction is constructive. To achieve this, we show that every constraint set from a given category can be used to express the constraints from the canonical set. Under the right restrictions, this immediately shows the reducibility from the canonical set to the other sets in the category. These restrictions are given in the notion of \emph{essential implementations} as defined below.
\begin{definition} \label{def:essential-implementation}
    Let~$f: \{0,1\}^p \to \{0,1\}$ be a constraint over~$p$ variables. A collection~$\C$ of constraint applications~$C_1, \ldots, C_m$ over variables~$\X = \{x_1, \ldots, x_p\}$ and dummy variables~$\Y = \{y_1, \ldots, y_q\}$ is an \emph{essential implementation} of the constraint application~$(f, \X)$ if:
    \begin{enumerate}[({I}1)]
        \item an assignment to~$\X$ satisfies~$(f,\X)$ if and only if there is an extension of this assignment to~$\X \cup \Y$ that satisfies all constraint applications~$C_1, \ldots, C_m$; and \label{itm:I1}
        \item for every assignment to $\X$, there is an extension to~$\X \cup \Y$ that satisfies~$C_2, \ldots,C_m$. (Note the exclusion of~$C_1$, which we call the \emph{head} of the implementation.) \label{itm:I2}
    \end{enumerate}
\end{definition}

We say that a constraint set~$\F$ \emph{essentially implements} a constraint~$f$ if~$f$ admits an essential implementation via constraint applications that are each from~$\F$. We denote this as~$\F \imp f$. We say that a constraint set~$\F_1$ essentially implements a constraint set~$\F_2$ if~$\F_1$ essentially implements each of the constraints from~$\F_2$. We denote this as~$\F_1 \imp \F_2$. In both arrow notations, if $\F_1$ is a singleton set, we may replace it by just its singleton element.

This notion of essential implementations is stricter than, e.g., the notion of pp-definitions, which are similarly defined but without requiring Property~\textbf{(I\ref{itm:I2})}. Such definitions facilitate reductions between formulas of different constraint sets that preserve the satisfiability of the formulas, as for example used by Schaefer~\cite{Schaefer78}. Even more similar is the notion of~$\alpha$-implementations as defined by Khanna et al.~\cite{KhannaSTW00}, especially those that are both `strict' and `perfect'. These implementations only differ from our \cref{def:essential-implementation} by allowing different assignments to~$\X$ that do not satisfy~$f$ to extend to an assignment on~$\X \cup \Y$ that leaves a different (but still singular) constraint application~$C_i$ unsatisfied. This uncertainty in which constraint application remains unsatisfied leads to changes in the structure of the solution space that makes these implementations unsuitable for our purposes. However, the authors use their definition of~$\alpha$-implementations to construct reductions that preserve the constant-factor approximability of \minF.

Our definition of essential implementations allows for a reduction that also preserves the exact approximation-ratio for constant factor-approximable instances. In fact, it even preserves the entire space of~$c$-approximate solutions in a predictable and useful manner. This is needed to prove the following result.
\begin{lemma} \label{lem:implementation-yields-reduction}
    Let~$\F_1$ and~$\F_2$ be constraint sets such that~$\F_1 \imp \F_2$, and let~$c \in \Rone$ be a constant. If \mincsp{$\F_1$} admits a polynomial-time $c$-essential detection algorithm, then \mincsp{$\F_2$} also admits a polynomial-time $c$-essential detection algorithm.
\end{lemma}
\begin{proof}
    We prove the statement by showing how to reduce an~$\F_2$-formula to an integer-weighted~$\F_1$-formula in such a way that a polynomial-time~$c$-essential detection algorithm on the~$\F_1$-formula can be used to detect~$c$-essential constraint applications in the~$\F_2$-formula. By \cref{lem:unweighted-to-weighted}, this suffices to prove the lemma. 
    
    Let~$F_2$ be an~$\F_2$-formula with constraint applications~$C_1, \ldots, C_\ell$ over variables~$\X = \{x_1, \ldots, x_n\}$. To transform~$F_2$ into a weighted~$\F_1$-formula $F_1$, we replace every~$C_i$ in~$F_2$ by an essential implementation~$\C_i$ of constraint applications from~$\F_1$ over variables $\X \cup \Y_i$. We denote the head of~$\C_i$ by~$C_i^\ast$. We give all constraint applications a weight of~$\lceil c \cdot \ell \rceil + 1$, except for the heads~$C_1^\ast, \ldots, C_\ell^\ast$ that each receive a weight of~$1$. We define $\Y:= \bigcup_{i\in[\ell]} \Y_i$. Since implementations are, by definition, of constant size, this results in a weighted $\F_1$-formula~$F_1$ of~$\Oh(\ell)$ constraint applications over the $\Oh(n+\ell)$ variables $\X \cup \Y$.

    To prove the correctness of this reduction, we start by proving the following claim about the non-head constraint applications in the resulting formula~$F_1$. To formulate the claim, we define~$H := \{ C_1^\ast, \ldots, C_\ell^\ast \}$ to be the set of all heads.

    \begin{claim} \label{clm:only-heads}
        Constraint applications in~$F_1$ that are not in~$H$ are not in any $c$-approximate solution of the resulting~\mincsp{$\F_1$} instance.
    \end{claim}
    \begin{claimproof}
        First, we show that~$H$ is a valid solution, i.e.,~$F_1 - H$ is satisfiable. We start with an arbitrary assignment~$\mathcal{X}$ to the variables~$\X$ and show that it can be extended to an assignment on the variables~$\X \cup \Y$ that satisfies~$F_1 - H$.

        Note that, for every~$i \in [\ell]$, due to Property~\textbf{(G\ref{itm:G2})}, the assignment~$\mathcal{X}$ can be extended to an assignment on~$\X \cup \Y_i$ that satisfies~$\C_i \setminus \{C_i^\ast\}$. Hence, the union of all these assignments satisfies all constraint applications in~$F_1$ except the heads. Hence,~$F_1 - H$ is satisfiable.

        Since all heads have a weight of~$1$, this yields that~$H$ is a solution with weight~$\ell$. All other constraint applications have a weight of~$\lceil c \cdot \ell \rceil + 1$, which is greater than~$c$ times the weight of solution~$H$. As such, no constraint application outside the set of heads~$H$ is in a $c$-approximate solution.
    \end{claimproof}

    Hence, every $c$-approximate solution of the resulting \mincsp{$\F_1$} instance is a subset of~$H$. Below, we show that the set~$H$ spans the exact same space of $c$-approximate solutions for~$F_1$ as the set of all constraint applications $C_1, \cdots, \C_\ell$ does for~$F_2$.

    \begin{claim}
    \label{clm:implementations-preserve-solution-space}
        Let~$S$ be a set of constraint applications in $F_2$ and let~$S_H \subseteq H$ be the corresponding set of heads in~$F_1$. Then,~$F_2 - S$ is satisfiable if and only if~$F_1 - S_H$ is satisfiable.
    \end{claim}
    \begin{claimproof}
        For one direction, let~$F_2 - S$ be satisfiable and let~$s$ be an assignment on~$\X$ that satisfies~$F_2 - S$. We show that an extension of~$s$ to~$\X \cup \Y$ exists that satisfies~$F_1 - S_H$.
        \begin{itemize}
            \item For every~$i \in [\ell]$ such that~$C_i \in S$, assignment~$s$ can be extended to an assignment on~$\X \cup \Y_i$ that satisfies all constraint applications in~$\C \setminus \{C_i^\ast\}$, due to Property~\textbf{(I\ref{itm:I2})}.
            \item For every~$i \in [\ell]$ such that~$C_i \notin S$, since~$s$ satisfies~$C_i$, assignment~$s$ can be extended to an assignment on~$\X \cup \Y_i$ that satisfies all constraint applications in~$\C_i$, due to Property~\textbf{(I\ref{itm:I1})}.
        \end{itemize}
        This covers all constraint applications in~$F_1-S_H$, showing the first direction of the claim.
        
        For the other direction, observe that any assignment~$s'$ on~$\X \cup \Y$ that satisfies~$F_1 - S_H$, satisfies the constraint applications in~$\C_i$ for every~$i$ such that~$C_i \notin S$. Thus, by Property~\textbf{(I\ref{itm:I1})}, the restriction of~$s'$ to~$\X$ satisfies~$F_2 - S$, thereby showing the other direction of the claim.
    \end{claimproof}
    Since every head has weight~$1$, it follows from the claims above that~$\opt(F_1) = \opt(F_2)$. Now we show how to use a polynomial-time~$c$-essential detection algorithm~$\mathcal{A}$ for~\wmincsp{$\F_1$} to construct a polynomial-time~$c$-essential detection algorithm for~\mincsp{$\F_2$}.

    Given an~$\F_2$-formula~$F_2$ and an integer~$k$, we start by transforming it into a weighted~$\F_1$-formula~$F_1$ using the reduction above; the weights are polynomial. Then, we apply~$\mathcal{A}$ to~$F_1$ and let~$D_1$ denote its output. If~$D_1$ contains non-head constraint applications of~$F_1$, we return the empty set. Otherwise, we return the set of constraint applications~$C_i$ for which the corresponding head~$C_i^\ast$ is in~$D_1$. In either case, let~$D_2$ be the output of our algorithm.

    We spend the remainder of the proof to show that the above is indeed a $c$-essential detection algorithm for \mincsp{$\F_2$}. To this end, suppose that~$\opt(F_2)=\opt(F_1) \leq k$. Then, since~$D_1$ satisfies Property~\textbf{(G\ref{itm:G1})}, there is a set~$S_1 \supseteq D_1$ such that~$S_1$ is an optimal solution to~$F_1$. By \cref{clm:only-heads}, $S_1$ only contains heads, and by extension so does $D_1$. As such, our proposed algorithm outputs the set~$D_2$ containing the constraint applications in~$F_2$ of which the heads of their essential implementation is in~$D_1$. For~$D_2$ to satisfy Property~\textbf{(G\ref{itm:G1})}, there must be a superset~$S_2 \supseteq D_2$ that is an optimal solution to~$F_2$. If we let~$S_2$ be the set of constraint applications in~$F_2$ of which the heads of their essential implementation are in~$D_2$, then this is indeed the case: by construction, it is a superset of~$D_2$ and because~$S_1$ is an optimal solution to~$F_1$, we have that~$S_2$ is an optimal solution to~$F_2$ by \cref{clm:implementations-preserve-solution-space}.

    To show that~$D_2$ also satisfies Property~\textbf{(G\ref{itm:G2})}, suppose that~$\opt(F_2) = \opt(F_1) = k$. Then, since~$D_1$ satisfies Property~\textbf{(G\ref{itm:G1})}, it contains all $c$-essential constraint applications from~$F_1$. Now, note that \cref{clm:only-heads,clm:implementations-preserve-solution-space} show that the entire space of $c$-approximations is preserved in our reduction with a one-to-one correspondence between constraint applications~$C_i$ in~$F_2$ and their respective head~$C_i^\ast$ in~$F_1$. As such, a constraint application~$C_i$ in~$F_2$ is $c$-essential if and only if its respective head~$C_i^\ast$ is $c$-essential in~$F_1$. Hence, since~$D_2$ is precisely the set of constraint applications in~$F_2$ to which those in~$D_1$ are the heads,~$D_2$ contains all $c$-essential constraint applications from~$F_2$.
\end{proof}
To aid us in finding essential implementations, we proceed by proving that the notion of essential implementations is a transitive relation.
\begin{lemma} \label{lem:transitivity}
    Let~$\F_1$,~$\F_2$, and~$\F_3$ be constraint sets. If~$\F_1 \imp \F_2$ and~$\F_2 \imp \F_3$, then~$\F_1 \imp \F_3$.
\end{lemma}
\begin{proof}
    Let~$(f_3, \X)$ be an arbitrary constraint application from~$\F_3$ over variables~$\X = \{x_1, \ldots, x_p\}$. To prove the lemma, we show that~$\F_1 \imp f_3$.

    Since~$\F_2 \imp \F_3$, there is an essential implementation of~$(f_3, \X)$ consisting of constraint applications~$C_1, \ldots, C_m$ from~$\F_2$ over variables~$\X$ and shared dummy variables~$\Y$. Let~$C_1$ be the head of this implementation. Since~$\F_1 \imp \F_2$, there is, for each~$i \in [m]$, an essential implementation~$\C_i$ of~$C_i$ consisting of constraint applications from~$\F_1$ over variables~$\X \cup \Y$ and dummy variables~$\Z_i$. Let~$C_i^\ast \in \C_i$ denote the head of~$\C_i$. Note that the constraint applications in one~$\C_i$ share one set of dummy variables~$\Z_i$, but that this set differs between different implementations~$\C_i$. We denote~$\Z := \bigcup_{i \in [m]} \Z_i$.

    We show that the concatenation of all constraint applications in~$\C_i, \ldots, \C_m$ is an essential implementation of~$f_3$ with head~$C_1^\ast$. We start by showing that it satisfies Property~\textbf{(I\ref{itm:I1})}.

    \begin{claim}
        An assignment to~$\X$ satisfies~$f_3$ if and only if there is an extension of this assignment to~$\X \cup \Y \cup \Z$ that satisfies all constraint applications in~$\C$.
    \end{claim}
    \begin{claimproof}
        By Property~\textbf{(I\ref{itm:I1})} of the essential implementation of~$C_1, \ldots, C_m$, an assignment to~$\X$ satisfies~$f_3$ if and only if there is an extension of it to~$\X \cup \Y$ that satisfies~$C_1, \ldots, C_m$.

        Likewise, an assignment to~$\X \cup \Y$ satisfies a constraint application~$C_i$ if and only if there is an extension of this assignment to~$\X \cup \Y \cup \Z_i$ that satisfies $\C_i$. By extension, an assignment to~$\X \cup \Y$ satisfies all constraints $C_1, \ldots, C_m$ if and only if there is an extension of this assignment to~$\X \cup \Y \cup \Z$ that satisfies all constraint applications in $\C$.

        The above two paragraphs combine to prove the claim.
    \end{claimproof}

    We continue by showing that the given implementation also satisfies Property~\textbf{(I\ref{itm:I2})}
    \begin{claim}
        For every assignment to~$\X$, there is an extension to~$\X \cup \Y \cup \Z$ that satisfies all constraint applications in~$\C \setminus C_1^\ast$.
    \end{claim}
    \begin{claimproof}
        Let~$s$ be an arbitrary assignment to~$\X$. By Property~\textbf{(I\ref{itm:I2})} of the implementation~$C_1, \ldots, C_m$, it extends to an assignment~$s'$ on~$\X \cup \Y$ that satisfies~$C_2, \ldots, C_m$. By Property~\textbf{(I\ref{itm:I1})} of the implementations~$\C_2, \ldots, \C_m$,~$s'$ extends to an assignment on~$\X \cup \Y \cup (\Z \setminus \Z_1)$ that satisfies all constraint applications in~$\C_2, \ldots, \C_m$. By Property~\textbf{(I\ref{itm:I2})} of the implementation~$\C_1$, every assignment on~$\X \cup \Y$ --- including~$s'$ --- extends to an assignment on~$\X \cup \Y \cup \Z_1$ that satisfies~$\C_1 \setminus\{C_1^\ast\}$. Thus,~$s$ extends to an assignment~$s'$ that in turn extends to an assignment on~$\X \cup \Y \cup \Z$ that satisfies~$\C \setminus \{C_1^\ast\}$, proving that~$\C$ satisfies Property~\textbf{(I\ref{itm:I2})}.
    \end{claimproof}
    This concludes the proof of the lemma.
\end{proof}
Before continuing with the hardness results, we prove some basic properties about essential implementations.
First, note for example that~$f \imp f$. It is also easy to see that if~$\F_1 \imp \F_2$ and~$\F_1 \subseteq \F_1'$, then~$\F_1' \imp \F_2$. We may use these results without reference. The following two also follow almost directly from the definition.
\begin{observation} \label{obs:quantified-implementation}
    Let~$f$ be a constraint over variables~$\{y_1, \ldots, y_q,x_1, \ldots, x_p\}$, and let $f'(x_1, \ldots, x_p) := \exists y_1, \ldots, y_q \mathrm{~s.t.~} f(y_1, \ldots, y_q,x_1, \ldots, x_p)$. Then,~$f \imp f'$.
\end{observation}
\begin{observation} \label{obs:substituted-implementation}
    Let~$f$ be a constraint over variable set~$\X$ and let~$V_1, \ldots, V_p$ be pairwise disjoint subsets of~$\X$. Then, for variables~$v_1, \ldots, v_p$ that are not necessarily in~$\X$, we have~$f \imp f\left[ \stack{V_1}{v_1} \comma \ldots \comma \stack{V_p}{v_p} \right]$.
\end{observation}
Next, we prove a property that involves the single-variable constraints $\true$ and $\false$.
\begin{lemma} \label{lem:implementation-with-constants}
    Let~$f$ be a constraint over variables~$\X := \{ x_1, \ldots, x_p \}$ and let $V = \{x_{i_1}, \ldots, x_{i_q}\} \subseteq \X$. Then,~$\{f, \false\} \imp f\left[ \stack{V}{0} \right]$ and~$\{f, \true\} \imp f\left[ \stack{V}{1} \right]$.
\end{lemma}
\begin{proof}
    We only provide a proof for the first claim. The second claim can be proven analogously. To this end, we show that $f\left[ \stack{x_{i_1}}{v_{i_1}} \comma \cdots \comma \stack{x_{i_q}}{v_{i_q}} \right] \wedge \false(v_{i_1}) \wedge \ldots \wedge \false(v_{i_q})$, with dummy variables~$v_{i_1}, \ldots, v_{i_q}$ and the first constraint application as its head, is an essential implementation of~$f\left[ \stack{Y}{0} \right]$.

    Because~$v_{i_1}, \ldots, v_{i_q}$ must be set to~$0$ to satisfy the last~$q$ constraint applications, it is easy to see that an assignment to~$\X$ can be extended to a satisfying assignment for our implementation if and only if it satisfies~$f\left[ \stack{Y}{0} \right]$. Thus, Property~\textbf{(I\ref{itm:I1})} is satisfied.

    Furthermore, by setting $v_{i_1} = \ldots = v_{i_q} = 0$, every assignment to~$\X$ can be extended to a satisfying assignment for~$\false(v_{i_1}) \wedge \ldots \wedge \false(v_{i_q})$. Thus, Property~\textbf{(I\ref{itm:I2})} is also satisfied.
\end{proof}

The remainder of \cref{sec:hardness} is used to prove the hardness results from \cref{thm:dichotomy}. In \cref{ssec:affine,ssec:wpos-wneg,ssec:other}, we prove Cases~\ref{itm:dichotomy-3}--\ref{itm:dichotomy-5} respectively.

\subsection{Affine constraint sets} \label{ssec:affine}
In this section, we show the hardness of~$c$-essential detection for~\minF for constraint sets~$\F$ that are affine but neither~$0$-valid,~$1$-valid, nor bijunctive. (Note that these conditions also imply that~$\F$ is not IHS-B). We pick~$\{\xor_4, \xnor_4\}$ as canonical constraint set for this category of constraint sets and prove in \cref{sssec:affine-canonical-hardness} that~$c$-essential detection for~\mincsp{$\{\xor_4, \xnor_4, \eqq\}$} is hard under the Unique Games Conjecture~(UGC). This conjecture poses that a specific version of \textsc{Label Cover} problem is NP-hard~\cite{Khot02} and has been used in the last two decades to prove many hardness of approximation results. Then, in \cref{sssec:affine-implementation}, we show for any affine constraint set~$\F$ that is neither $0$-valid, nor $1$-valid, nor bijunctive, that~$\F \imp \{ \xor_4, \xnor_4, \eqq \}$.

\subsubsection{Hardness of c-essential detection for the canonical problem} \label{sssec:affine-canonical-hardness}
We give a reduction from~\mincsp{$\{\eqq, \neqq\}$}, for which the following known hardness result accompanied the introduction of the UGC.

\begin{lemma}[{\cite[Theorem 3]{Khot02}}] \label{lem:min-uncut-hardness}
    Let~$\frac12<t<1$ and sufficiently small~$\eps>0$ be constants. Unless P=NP or the UGC fails, there is no polynomial-time algorithm that distinguishes, for a given~\mincsp{$\{\eqq, \neqq\}$}-instance~$F$ with~$m$ constraints, between the following two cases:
    \begin{enumerate}[(i)]
        \item $\opt(F) \leq \eps \cdot m$;
        \item $\opt(F) \geq \eps^t \cdot m$.
    \end{enumerate}
\end{lemma}

Below, we use this result as starting point for our hardness reduction.

\begin{lemma} \label{lem:affine-canonical-hardness}
    Unless P=NP or the UGC fails, there is no constant~$c\in \Rone$ for which a polynomial-time~$c$-essential detection algorithm exists for~\mincsp{$\{\xor_4, \xnor_4, \eqq\}$}.
\end{lemma}
\begin{proof}
    Let~$c \in \Rone$ be an arbitrary constant. We provide a reduction from the problem posed to be hard in \cref{lem:min-uncut-hardness}, where we pick~$t = \frac34$ and~$\eps > 0$ to be small enough to satisfy~$\eps^{t-1} = \eps^{-\frac14} > 2c$. To this end, let~$F$ be an~$\{\eqq, \neqq\}$-formula that we assume w.l.o.g. to have~$m \geq \eps\inv$ constraints. Now, we show how to transform~$F$ into an integer-weighted~$\{\xor_4, \xnor_4, \eqq\}$-formula~$F'$ and argue that the results of a~$c$-essential detection algorithm on~$F'$ can be used to solve the distinguishing task from \cref{lem:min-uncut-hardness}. By \cref{lem:unweighted-to-weighted}, this suffices to prove the lemma.

    We define the formula~$F'$ over the same set of variables as~$F$ with three variables~$z_1, z_2, z_3$ added to it. We start the construction of~$F'$ by adding the constraint application~$(z_1 = z_2)$ to it with weight~$c \cdot (\lfloor\eps \cdot m\rfloor+1)+1$ and adding the constraint application~$(z_2 = z_3)$ with weight~$\lfloor\eps \cdot m\rfloor+1$. Then, for every constraint application of an equality~$(x = y)$ in~$F$, we add the constraint application~$\xnor_4(z_1, z_2, x, y)$ to~$F'$ with weight~$1$. Finally, for every inequality constraint application~$(x \neq y)$ in~$F$, we add the constraint application~$\xor_4(z_2, z_3, x, y)$ to~$F'$ with weight~$1$.

    To argue the correctness of this reduction, we start by proving the following claim.

    \begin{claim} \label{clm:eq-neq-to-nor-xor-solution}
        A set $X$ of constraint applications from~$F$ is such that~$F-X$ is satisfiable if and only if the set~$X'$ of corresponding~$\xor_4$ and~$\xnor_4$ constraint applications from~$F'$ is such that~$F'-X'$ is satisfiable.
    \end{claim}
    \begin{claimproof}
        Consider an equality constraint application~$(x=y)$ from~$F$ and note that, if~$z_1 = z_2$, an assignment to~$x$ and~$y$ satisfies~$(x = y)$ if and only if it satisfies the corresponding constraint application~$\xnor_4(z_1, z_2, x, y)$. Similarly, consider an inequality constraint~$(x \neq y)$ from~$F$ and note that, if~$z_2 = z_3$, an assignment to~$x$ and~$y$ satisfies~$(x \neq y)$ if and only if it satisfies the corresponding constraint application~$\xor_4(z_2, z_3, x, y)$.

        Thus, for a set~$X$ of constraint applications from~$F$ and the set~$X'$ of corresponding~$\xor_4$ and~$\xnor_4$ constraint applications, an assignment that satisfies~$F - X$ can be extended to one that satisfies~$F'-X'$ by setting~$z_1$,~$z_2$ and~$z_3$ to~$0$. Likewise, the restriction of an assignment that satisfies~$F'- X'$ to the variables in~$F-X$ will satisfy~$F-X$, since it satisfies the constraint applications~$(z_1 = z_2)$ and~$(z_2 = z_3)$ from~$F' - X'$ in particular.
    \end{claimproof}
    We use this claim above to prove that the result of a~$c$-essential detection algorithm for \mincsp{$\{\xor_4, \xnor_4, \eqq\}$} on~$F'$ can be used to determine whether the original formula~$F$ satisfies~$\opt(F) \leq \eps \cdot m$ or~$\opt(F) > \eps^t \cdot m$. To this end, let~$S$ be the result of a~$c$-essential detection algorithm for~\mincsp{$\{\xor_4, \xnor_4, \eqq\}$} on input~$F'$ with~$k=\lfloor\eps \cdot m\rfloor+1$. Now, to recognize the first case, consider the following.

    \begin{claim}
        If~$\opt(F) \leq \eps \cdot m$, then~$S$ does not contain the constraint application~$(z_2 = z_3)$.
    \end{claim}
    \begin{claimproof}
        Let~$X$ be a set of constraint applications such that~$F-X$ is satisfiable and~$|X| \leq \eps \cdot m$. Then, by \cref{clm:eq-neq-to-nor-xor-solution}, there is a solution to~$F'$ of weight~$\leq \eps \cdot m$. Thus,~$\opt(F') \leq \eps \cdot m$ is upper bounded by~$k=\lfloor\eps \cdot m\rfloor+1$. As such,~$S$ must satisfy Property~\textbf{(G\ref{itm:G1})} and be a subset of an optimal solution to~$F'$. Since there is a solution to~$F'$ with weight~$\leq \eps \cdot m$ and the constraint application~$(z_2 = z_3)$ has the strictly larger weight~$\lfloor\eps \cdot m\rfloor+1$, this constraint application is not in any optimal solution and therefore not in~$S$.
    \end{claimproof}

    Next, to recognize the second case, consider the following claim.

    \begin{claim}
        If~$\opt(F) \geq \eps^t \cdot m$, then~$S$ does contain the constraint application~$(z_2 = z_3)$.
    \end{claim}
    \begin{claimproof}
        First, we argue that the singleton set containing just the constraint application~$(z_2 = z_3)$ of weight~$\lfloor\eps \cdot m\rfloor+1$ is a solution to~$F'$. To see this, one can verify that all other constraint applications in~$F'$ are satisfied by setting~$z_3$ to~$1$ and all other variables to~$0$. We continue by showing that all solutions that do not contain~$(z_2 = z_3)$ must have a weight that is more than~$c \cdot (\lfloor\eps \cdot m\rfloor+1)$.

        Let~$X'$ be an arbitrary solution to~$F'$ that does not contain the constraint application~$(z_2 = z_3)$ and suppose for contradiction that it has a weight of at most~$c \cdot (\lfloor\eps \cdot m\rfloor+1)$. Since the constraint application~$(z_1 = z_2)$ individually exceeds this weight by having a weight of~$c \cdot (\lfloor\eps \cdot m\rfloor+1)+1$, we conclude that it is not in~$X'$. As such,~$X'$ must consist entirely of weight-$1$ constraint applications.

        However, by \cref{clm:eq-neq-to-nor-xor-solution}, a solution to~$F'$ that only consists of such weight-$1$ constraint applications implies that~$F$ has a solution of equal size. By assumption,~$F$ only contains solutions of size at least~$\eps^t \cdot m$. As such, the weight of~$X'$ must also be at least~$\eps^t \cdot m$.

        Now, we can use our assumptions on~$\eps$ and~$m$ to derive that the weight of~$X'$ is at least~$\eps^t \cdot m = \eps^{t-1}\cdot \eps \cdot m > 2c\cdot\eps \cdot m \geq c \cdot (\eps \cdot m + 1) \geq c \cdot (\lfloor\eps \cdot m\rfloor + 1)$. The strict inequality in this derivation follows from our choice of~$\eps$ that satisfies~$\eps^{t-1} > 2c$. The inequality directly afterwards follows from the assumption that $m \geq \eps\inv$ so that~$\eps \cdot m \geq 1$. This contradicts our assumption that~$X'$ has weight at most~$c \cdot (\lfloor\eps \cdot m\rfloor+1)$.

        We conclude that all solutions to~$F'$ that do not contain~$(z_2 = z_3)$ are more than~$c$ times as heavy as the constraint application~$(z_2 = z_3)$ itself. This means that~$\{(z_2 = z_3)\}$ is an optimal solution of weight~$\lceil\eps \cdot m\rceil +1$ and even that~$(z_2 = z_3)$ is~$c$-essential. As~$k=\lceil\eps \cdot m\rceil +1=\opt(F)$,~$S$ must contain all~$c$-essential constraint applications to satisfy Property~\textbf{(G\ref{itm:G2})}. Hence,~$S$ contains~$(z_2 = z_3)$.
    \end{claimproof}
    Now, the two claims above indicate how a polynomial-time~$c$-essential detection algorithm for~$F'$ can be used to distinguish between~$\opt(F) \leq \eps \cdot m$ and~$\opt(F) > \eps^t \cdot m$ in polynomial time: after running such an algorithm, it suffices to check whether the output contains~$(z_2 = z_3)$. By \cref{lem:min-uncut-hardness}, this distinction is NP-hard to make. As~$c\in\Rone$ was chosen arbitrarily, this shows that~$c$-essential detection for~\wmincsp{$\{\xor_4, \xnor_4, \eqq\}$} is NP-hard for every~$c \in \Rone$. By \cref{lem:unweighted-to-weighted}, the same holds for the unweighted variant of the problem.
\end{proof}

As suggested by a reviewer, it is likely possible to write a similar proof to the one above, instead reducing from \mincsp{$\{\xor_3, \xnor_3\}$}. Then, using a known NP-hardness result for this problem in place of our current \cref{lem:min-uncut-hardness}, it would be possible to have \cref{lem:affine-canonical-hardness} depend solely on the assumption that P$\neq$NP and no longer on the UGC.

\subsubsection{Essential implementation of the canonical constraint set}
\label{sssec:affine-implementation}
The main result of this section is shown in \cref{lem:affine-implementation}. In it, we require several essential implementations involving~$\xor$ and~$\xnor$ constraints. We start by showing the following.

\begin{lemma} \label{lem:larger-xor-xnor}
    For all~$p \geq 2$,~$\xor_p \imp \xnor_{2(p-1)}$ and~$\xnor_p \imp \xnor_{2(p-1)}$.
\end{lemma}
\begin{proof}
    We show that $\xor_p(x_1, \ldots, \linebreak[1] x_{p-1}, v) \wedge \xor_p(x_p, \ldots,x_{2(p-1)},v)$ and \\ $\xnor_p(x_1, \ldots, x_{p-1}, v) \wedge \xnor_p(x_p, \ldots,x_{2(p-1)},v)$ are both essential implementations of~$\xnor_{2(p-1)}(x_1, \ldots, x_{2(p-1)})$, with~$v$ as dummy variable and their first constraint application as head. Observe that for both pairs of constraint applications, summing them modulo two yields the~$\xnor$ constraint on $2(p-1)$ variables. As such, both implementations satisfy Property~\textbf{(I\ref{itm:I1})}. To see that they satisfy Property~\textbf{(I\ref{itm:I2})}, observe that any assignment to~$x_p, \ldots,x_{2(p-1)}$ can be extended to a satisfying assignment for the second constraint application by picking the correct value for~$v$.
\end{proof}

Next, it is easy to see that for dummy variable~$v$ we have~$\xor_{p-2}(x_1, \ldots,x_{p-2}) \equiv \xor_p(v,v,x_1, \ldots, x_{p-2})$ and~$\xnor_{p-2}(x_1, \ldots,x_{p-2}) \equiv \xnor_p(v,v,x_1, \ldots, x_{p-2})$ and that both actually provide essential implementations, leading to the following result.

\begin{observation} \label{obs:smaller-xor-xnor}
    For all~$p \geq 3$,~$\xor_p \imp \xor_{p-2}$ and~$\xnor_p \imp \xnor_{p-2}$.
\end{observation}

Below are four more examples of simple essential implementations that will be used in the proof of \cref{lem:affine-implementation}. We give the construction of each of these implementations and leave the very straightforward task of verifying that they satisfy Properties~\textbf{(I\ref{itm:I1})} and~\textbf{(I\ref{itm:I2})} to the reader.

\begin{lemma} \label{lem:xor2-xnor6-imps-xor4}
    $\{\xor_2, \xnor_6\} \imp \xor_4$.
\end{lemma}
\begin{proof}
    $\xnor_6(v_1,v_2,x_1,x_2,x_3,x_4) \wedge \xor_2(v_1,v_2)$ with dummy variables~$v_1$ and~$v_2$ and the first constraint application as head is an essential implementation of~$\xor_4(x_1,x_2,x_3,x_4)$.
\end{proof}
\begin{lemma} \label{lem:odd-xnor-imps-false}
    If~$p$ is an odd integer, then~$\xnor_p \imp \false$.
\end{lemma}
\begin{proof}
    $\xnor_p(x, \ldots,x)$ is an essential implementation of~$\false(x)$.
\end{proof}
\begin{lemma} \label{lem:xnor6-true-imps-xor5}
    $\{ \xnor_6, \true\} \imp \xor_5$.
\end{lemma}
\begin{proof}
    $\xnor_6(v,x_1,x_2,x_3,x_4,x_5) \wedge \true(v)$ with dummy variable~$v$ and the first constraint application as head is an essential implementation of~$\xor_5(x_1,x_2,x_3,x_4,x_5)$.
\end{proof}
\begin{lemma} \label{lem:xor5-false-imps-xor4}
    $\{\xor_5, \false\} \imp \xor_4$
\end{lemma}
\begin{proof}
    $\xor_5(v, x_1, x_2,x_3,x_4) \wedge \false(v)$ with dummy variable~$v$ and the first constraint application as head is an essential implementation of~$\xor_4(x_1,x_2,x_3,x_4)$.
\end{proof}
Now, we can use the above results to show our desired implementation
\begin{lemma} \label{lem:affine-implementation}
    Let~$\F$ be a constraint set that is neither~$0$-valid, nor~$1$-valid, nor bijunctive. If~$\F$ is affine, then~$\F \imp\{\xor_4, \xnor_4, \eqq\}$.
\end{lemma}
\begin{proof}
    We start by using an insight from the proof of Lemma 4.18 in the previous work of Khanna et al.~\cite{KhannaSTW00}. In it, they define, for a constraint~$f$ over variable set~$\X$, a subset~$S \subseteq \X$ as a \emph{dependent set} of variables if there is at least one assignment to~$S$ that does not extend to a satisfying assignment for~$f$. A dependent set~$S$ is called a \emph{minimally dependent set} of~$f$ if no strict subset of it is also a dependent set of~$f$. For a set~$S \subseteq \X$, let~$f_S$ over variable set~$S$ be defined as~$f_S :=\exists \X \setminus S \mathrm{~s.t.~} f(\X)$. By \cref{obs:quantified-implementation},~$f \imp f_S$ for any~$S \subseteq \X$.

    Now, the proof by Khanna et al.\ shows that every affine function~$f$ can be written as the conjunction of~$f_S$ over all minimally dependent sets~$S$ of~$f$. Moreover, they show that for every minimally dependent set~$S$, the constraint~$f_S$ can be described by a single linear Boolean equation. That is,~$f_S$ is equivalent to either~$\xor_{|S|}$ or~$\xnor_{|S|}$. This shows that~$f$ can be written as a conjunction of linear Boolean equations that are each even essentially implemented by~$f$.

    To start our chain of implementations, we start by observing that~$\F$ must contain some~$f$ with at least one minimally dependent set of size~$3$, as~$\F$ would be bijunctive otherwise. Thus, there is some~$p \geq 3$ such that~$\F \imp \xor_p$ or~$\F \imp \xnor_p$. In either case, we can apply \cref{lem:larger-xor-xnor} to see that~$\F \imp \xnor_{2(p-1)}$. Since~$2(p-1)$ is an even integer of at least~$4$, we can repeatedly apply \cref{obs:smaller-xor-xnor} to derive that~$\F \imp \xnor_4$ and~$\F \imp \eqq$ (since~$\eqq$ is~$\xnor_2$). It remains to show that~$\F \imp \xor_4$.

    We proceed by combining the fact that~$\F \imp \xnor_4$ with \cref{lem:larger-xor-xnor} to conclude that~$\F \imp \xnor_6$. Next, we note that~$\F$ cannot be expressed as a collection of conjunctions of~$\xnor$ constraints only, as this would make~$\F$~$0$-valid, which we assumed it not to be. Thus,~$\F$ must contain a constraint $f$ with a minimally dependent set~$S$ such that~$f_S$ is equivalent to~$\xor_q$ for some integer~$q \geq 1$. In turn,~$\F \imp \xor_q$. For the remainder of this proof, we distinguish two cases.

    For the first and simplest case, we assume~$q$ to be even. Then, repeatedly applying \cref{obs:smaller-xor-xnor} yields that~$\xor_q \imp \xor_2$. It follows from \cref{lem:xor2-xnor6-imps-xor4} that~$\F \imp \xor_4$, thereby proving the required statement for this case.

    For the second case, we assume~$q$ to be odd. Then, repeatedly applying \cref{obs:smaller-xor-xnor} yields that~$\xor_q \imp \true$ (since $\true$ is $\xor_1$). Next, we make the observation that~$\F$ cannot be expressed as a collection of conjunctions of only XOR-constraints over odd numbers of variables and XNOR-constraints over even numbers of variables. Such conjunctions are~$1$-valid, which we assumed~$\F$ not to be. Thus,~$\F$ must implement at least one XOR-constraint over an even number of variables or at least one XNOR-constraint over an odd number of variables. If the former is true, we can resolve the proof using the case presented in the previous paragraph. Thus, we assume that~$f \imp \xnor_r$ for some odd integer~$r$.

    Now, it follows from \cref{lem:odd-xnor-imps-false} that~$\xnor_r \imp \false$, which establishes that~$f \imp \{\true, \false\}$ in this case. We use this to apply \cref{lem:xnor6-true-imps-xor5} and see that~$f \imp\xor_5$ after which it follows from \cref{lem:xor5-false-imps-xor4} that~$f \imp \xor_4$.

    Hence, we conclude that in both cases~$\F \imp \{\xor_4, \xnor_4, \eqq\}$.
\end{proof}

Now, \cref{lem:affine-canonical-hardness,lem:affine-implementation} combine to prove Case~\ref{itm:dichotomy-3} of \cref{thm:dichotomy}.

\subsection{Weakly positive and weakly negative constraint sets}
\label{ssec:wpos-wneg}
In this section, we prove Case~\ref{itm:dichotomy-4} of \cref{thm:dichotomy}. We pick~$\{\orr_{3,1}, \orr_{2,1}, \true, \linebreak[1] \false\}$ as canonical constraint set for weakly positive constraint sets and we pick~$\{\orr_{3,2}, \orr_{2,1}, \true, \linebreak[1] \false\}$ as canonical constraint set for weakly negative constraint sets. First, we show for each canonical constraint set~$\Fcan$ that~\textsc{$c$-essential Detection for~\mincsp{$\Fcan$}} is NP-hard in \cref{sssec:wpos-wneg-canonical-hardness}. Then, in \cref{sssec:wpos-wneg-implementation}, we show for any weakly positive (resp. weakly negative) constraint set~$\F$ that is neither $0$-valid, $1$-valid, bijunctive, nor IHS-B, that~$\F \imp \Fcan$ for the corresponding canonical constraint set~$\Fcan$. (Note that these conditions imply~$\F$ is not affine either.)

\subsubsection{Hardness of c-essential detection for the canonical problems}
\label{sssec:wpos-wneg-canonical-hardness}
We give a reduction from the \textsc{Set Cover} optimization problem. In this problem, a universe~$U$ is given together with a collection~$\mathcal{S}$ of subsets of~$U$, and the goal is to determine the size of a smallest subset~$Y \subseteq \mathcal{S}$ such that every $u \in U$ is contained in at least one set in~$Y$. We call such a set~$Y$ (of any size) a \emph{set cover} of~$(U, \mathcal{S})$, and we denote the size of a smallest set cover of~$(U, \mathcal{S})$ by~$\opt(U, \mathcal{S})$. The following hardness result is known.

\begin{lemma}[{\cite[Theorem 22.31]{AroraB09}}]
\label{lem:hitting-set-hardness}
    For every sufficiently small constant~$\eps > 0$, the following holds. Unless P=NP, there is no polynomial-time algorithm that takes an integer~$t \geq 1$ and a \textsc{Set Cover} instance~$(U, \mathcal{S})$ as input and distinguishes between the following two cases:
    \begin{enumerate}[(i)]
        \item $\opt(U, \mathcal{S}) \leq t$;
        \item $\opt(U, \mathcal{S}) > \frac{t}{4\sqrt\eps}$.
    \end{enumerate}
\end{lemma}
 
We use this result to prove the statement below.

\begin{lemma} \label{lem:wp-canonical-hardness}
    Unless P=NP, there is no constant~$c \in \Rone$ for which a polynomial-time $c$-essential detection algorithm exists for \mincsp{$\{\orr_{3,1}, \orr_{2,1}, \true, \false\}$} or for \mincsp{$\{\orr_{3,2}, \orr_{2,1}, \true, \false\}$}.
\end{lemma}
\begin{proof}
    We only prove that~$c$-essential detection for \mincsp{$\{\orr_{3,1}, \orr_{2,1}, \true, \false\}$} is hard. The proof that the same holds for \mincsp{$\{\orr_{3,2}, \orr_{2,1}, \true, \false\}$} follows analogously.

    Now, let $\horn_p := \left( \bigcup_{i \in \{2 \ldots p\}} \orr_{i,1} \right) \cup\{\true, \false\}$, so that the lemma states hardness of $c$-essential detection for \mincsp{$\horn_3$}. We start by showing that the existence of a polynomial-time~$c$-essential detection algorithm for \wmincsp{$\horn_\infty$} implies that P=NP. We assume that such an algorithm exists for an arbitrary value of~$c\in\Rone$ and argue that it can be used to solve the distinguishing task from \cref{lem:hitting-set-hardness} in polynomial time. Afterwards, we explain how to extend this result to~$\horn_3$.

    Now, let~$(U, \mathcal{S})$ be a \textsc{Set Cover} instance, and let~$t \geq 1$ be an arbitrary integer. We let~$m$ denote the maximum size of a set in~$\mathcal{S}$, and show how to transform~$(U, \mathcal{S})$ into a weighted $\horn_{m+1}$-formula~$F$.
     
    We start by introducing two variables~$z$ and~$w$ and adding constraint applications~$\true(w)$ with weight~$c(t+1)+1$ and~$\false(z)$ with weight~$t+1$ to~$F$. Then, for every set~$S_j \in \mathcal{S}$, we introduce a variable~$y_j$ and we add a constraint application~$(\neg y_j \vee z)$ to~$F$ with weight~$1$.

    Next, for every~$u_i \in U$, we introduce a variable~$x_i$ and we add~$(x_i \vee \neg w)$ to~$F$ with weight~$c(t+1) + 1$. Additionally, if we let~$S_{j_1}, \ldots, S_{j_\ell}$ be the collection of sets that~$u_i$ appears in, we add~$(\neg x_i \vee y_{j_1} \vee \ldots \vee y_{j_\ell})$ to~$F$ with weight~$c(t+1)+1$.

    This concludes the reduction. Observe that the resulting formula~$F$ is a $\horn_{m+1}$-formula whose size is polynomial in the total size of the original \textsc{Set Cover} instance. We continue by proving the correctness of our reduction, starting with the following claim.
    
    \begin{claim} \label{clm:hitting-set-solution}
        A collection~$\C := S_{j_1}, \ldots, S_{j_r}$ is a set cover of~$(U, \mathcal{S})$ if and only if the collection of constraint applications~$X_\C := \{(\neg y_{j_1} \vee z), \ldots, (\neg y_{j_r} \vee z)\}$ is such that~$F-X_\C$ is satisfiable.
    \end{claim}
    \begin{claimproof}
        For one direction, suppose that~$\C$ is a set cover of~$(U, \mathcal{S})$. Next, consider the assignment~$s$ to the variables of~$F$ that sets~$w$ to~$1$,~$z$ to~$0$, all variables~$x_i$ to~$1$, the variables~$y_j$ for which~$S_j \in \C$ to~$1$, and the variables~$y_j$ for which~$S_j \notin \C$ to~$0$.

        Clearly,~$\true(w)$ and~$\false(z)$ are satisfied by this assignment. All constraint applications of the form~$(\neg y_j \vee z)$ in~$F-X_\C$ are also satisfied because all variables~$y_j$ for which~$S_j \notin \C$ are set to~$0$. Finally, note that, because~$\C$ is a set cover of~$U$, every constraint application of the form~$(\neg x_i \vee y_{j_1} \vee \ldots \vee y_{j_\ell})$ contains at least one~$y_j$ for which~$S_j \in \C$. All such variables are set to~$1$, so all constraint applications of this form in~$F-X_\C$ are also satisfied by~$s$.

        For the other direction, suppose that~$F-X_\C$ is satisfiable and let~$s$ be a satisfying assignment. To satisfy~$\true(w)$ and~$\false(z)$, we must have~$s(w) = 1$ and~$s(z) = 0$. Then, to satisfy all constraint applications of the form~$(x_i \vee \neg w)$, we must have~$s(x_i) = 1$ for all variables~$x_i$. Likewise, to satisfy all constraint applications in~$F-X_\C$ of the form~$(\neg y_j \vee z)$, we must have~$s(y_j)=0$ for all variables~$y_j$ with~$S_j \notin \C$.

        Next, consider an arbitrary constraint application of the form~$(\neg x_i \vee y_{j_1} \vee \ldots \vee y_{j_\ell})$. Since~$s$ satisfies this constraint application but also has~$s(x_i) = 1$ it must set at least one of the variables~$y_{j_1}, \ldots, y_{j_\ell}$ to~$1$. However, since~$s$ sets all variables~$y_j$ for which~$S_j \notin \C$ to~$0$, at least one of the variables~$y_{j_1}, \ldots, y_{j_\ell}$ must correspond to a set from~$\mathcal{S}$ that is in~$\C$. Thus, we obtain that for every variable~$x_i$, the constraint application~$(\neg x_i \vee y_{j_1} \vee \ldots \vee y_{j_\ell})$ contains at least one variable~$y_j$ for which~$S_j \in \C$. This means that~$\C$ is a set cover of~$(U, \mathcal{S})$.
    \end{claimproof}
    Now, we define~$\eps := \frac{1}{65c^2}$ and use the claim above to show that a $c$-essential detection algorithm for \mincsp{$\horn_p$} on input~$F$ can be used to determine whether the original \textsc{Set Cover} input~$(U, \mathcal{S})$ satisfies~$\opt(U, \mathcal{S}) \leq t$ or~$\opt(U, \mathcal{S}) > \frac{t}{4\sqrt\eps}$. To this end, let~$S$ be the result of a $c$-essential detection algorithm for \mincsp{$\horn_p$} on input~$F$ with~$k = t+1$. To recognize the first case, consider the following claim.

    \begin{claim}
        If~$\opt(U, \mathcal{S}) \leq t$, then~$S$ does not contain the constraint application~$\false(z)$.
    \end{claim}
    \begin{claimproof}
        Let~$\C \subseteq \mathcal{S}$ be a set cover of~$(U, \mathcal{S})$ of size at most~$t$. Then, by \cref{clm:hitting-set-solution}, the corresponding set~$X_\C$ is such that~$F-X_\C$ is satisfiable. Since all constraint applications in~$X_\C$ have weight~$1$, this constraint set is a solution to~$F$ of weight~$t$. Thus,~$\opt(F) \leq t < k$. As such,~$S$ must be a subset of an optimal solution to~$F$ to satisfy Property~\textbf{(G\ref{itm:G1})}. Since an optimal solution of~$F$ has weight~$\leq t$ and the constraint application~$\false(z)$ has weight~$t+1$,~$\false(z)$ is not in any optimal solution and by extension also not in~$S$.
    \end{claimproof}
    We continue by showing how the set~$S$ differs for the other case.
    \begin{claim}
        If~$\opt(U, \mathcal{S}) > \frac{t}{4\sqrt\eps}$, then~$S$ does contain the constraint application~$\false(z)$.
    \end{claim}
    \begin{claimproof}
        First, we argue that the singleton set~$\{\false(z)\}$ of weight~$t+1$ is a solution to~$F$. To see this, note that every constraint application in~$F-\{\false(z)\}$ contains at least one positive literal so that the all-one assignment satisfies it. We continue by showing that solutions to~$F$ that do not contain~$\false(z)$ have a weight of more than~$c\cdot(t+1)$.

        Let~$X$ be an arbitrary such solution and suppose for contradiction that it has a weight of~$\leq c \cdot (t+1)$. Recall that the only constraint applications that do not individually exceed this weight are~$\false(z)$ and the weight-$1$ constraint applications of the type~$(\neg y_i \vee z)$. As such,~$X$ consists solely of such weight-$1$ constraints. 
        
        However, by \cref{clm:hitting-set-solution}, a solution that only consists of these weight-$1$ constraint applications corresponds to a set cover of~$(U, \mathcal{S})$ of equal size. Therefore,~$X$ must be at least of size~$\opt(U, \mathcal{S})$, which, by assumption, is larger than~$\frac{t}{4\sqrt\eps}$. We defined~$\eps = \frac{1}{65c^2}$, so by rewriting the inequality~$\eps < \frac{1}{64c^2}$, we obtain that~$\frac{1}{4\sqrt\eps} > 2c$, which in turn yields that~$\frac{t}{4\sqrt\eps} > 2ct$. Since $2t \geq t+1$ for all $t \geq 1$, we get that~$|X| \geq \opt(U, \mathcal{S}) > \frac{t}{4\sqrt\eps} > c \cdot (t+1)$. This contradicts our assumption that~$X$ has weight at most~$c \cdot (t+1)$.

        We conclude that all solutions to~$F$ that do not contain~$\false(z)$ are more than~$c$ times as heavy as~$\false(z)$ itself. This means that~$\{\false(z)\}$ is an optimal solution of weight~$(t+1)$ and even that~$\false(z)$ is~$c$-essential. As~$k = t+1 = \opt(F)$,~$S$ must contain all~$c$-essential constraint applications to satisfy Property~\textbf{(G\ref{itm:G1})}. Hence,~$S$ contains~$\false(z)$.
    \end{claimproof}
    Now, the two claims above indicate how a polynomial-time $c$-essential detection algorithm for~$F$ can be used to distinguish between~$\opt(U, \mathcal{S}) \leq t$ and~$\opt(U, \mathcal{S}) > \frac{t}{4\sqrt\eps}$ in polynomial time as well: after running such an algorithm, it suffices to check whether the output contains~$\false(z)$. By \cref{lem:hitting-set-hardness}, this distinction is NP-hard to make. As~$c\in\Rone$ was chosen arbitrarily, this shows that $c$-essential detection for~\wmincsp{$\horn_\infty$} is NP-hard for every~$c \in \Rone$. By \cref{lem:unweighted-to-weighted}, the same holds for the unweighted variant of the problem.

    It remains to show that the same is true for \mincsp{$\horn_3$}. To this end, we show how to transform the~$\horn_{m+1}$-formula~$F$ into a~$\horn_3$-formula whose size is polynomial in the size of~$F$. We note that, for any integer~$p \geq 2$, $\{ \orr_{2,1}, \orr_{3,1} \} \imp \orr_{p,1}$. This can be seen by inductively applying the following claim.
    \begin{claim} \label{clm:or-p-imps-or-p+1}
        Let~$p \geq 4$. Then $\{ \orr_{p-1,1}, \orr_{3,1} \} \imp \orr_{p,1}$.
    \end{claim}
    \begin{claimproof}
        We show that~$(\neg x_1 \vee x_2 \vee \ldots \vee x_{p-2} \vee v) \wedge (\neg v \vee x_{p-1} \vee x_p)$, with dummy variable~$v$ and the first constraint application as its head, is an essential implementation of~$(\neg x_1 \vee x_2 \vee \ldots \vee x_p)$. It is easily seen that this implementation satisfies Property~\textbf{(I\ref{itm:I1})} by verifying that $(\neg x_1 \vee x_2 \vee \ldots \vee x_{p-2} \vee v) \wedge (\neg v \vee x_{p-1} \vee x_p)$ is satisfied if and only if~$(\neg x_1 \vee x_2 \vee \ldots \vee x_p)$ is satisfied.

        Next, note that every assignment to the variables~$x_1, \ldots, x_p$ extends to a satisfying assignment for $(\neg v \vee x_{p-1} \vee x_p)$ by setting~$v$ to~$0$. Hence, Property~\textbf{(I\ref{itm:I2})} is satisfied.
    \end{claimproof}
    Now, a single essential implementation of~$\orr_{p,1}$ can be made using~$\Oh(p)$ applications of~$\orr_{2,1}$ and~$\orr_{3,1}$, since the construction in \cref{clm:or-p-imps-or-p+1} shows that we can build an~$\orr_{p,1}$ application from one (constant-size)~$\orr_{3,1}$ application and only one~$\orr_{p-1,1}$ application. 
    
    To transform~$F$ into a~$\horn_3$-formula, we replace every~$\horn_p$ application in it with~$p > 3$ by an essential implementation of~$\orr_{2,1}$ and~$\orr_{3,1}$ applications in which the head receives weight~$1$ and the other applications receive weight~$c(t+1)+1$. Then, the same logic as in \cref{lem:implementation-yields-reduction} shows that a polynomial-time $c$-essential detection algorithm for the resulting~$\horn_3$-formula can be used to perform~$c$-essential detection for the original formula~$F$ in polynomial time. This, in turn, was shown to imply P=NP.
\end{proof}

\subsubsection{Essential implementation of the canonical constraint set}
\label{sssec:wpos-wneg-implementation}
We start by giving some additional background. First, we introduce the notion of a maxterm.

\begin{definition}
    Let~$f$ be a constraint over~$\X$ and let~$s$ be an assignment to~$\X' \subseteq \X$ that fixes~$f$ to~$0$. (That is, no extension of~$s$ to~$\X$ satisfies~$f$). If no restriction of~$s$ to a strict subset of~$\X'$ also fixes~$f$ to~$0$, we call~$s$ a \emph{maxterm} of~$f$.
\end{definition}

We represent a maxterm~$s$ over~$\X'$ by a CNF clause with, for every variable~$x \in \X$, the literal that~$s$ does not satisfy. E.g. if the assignment~$s$ is defined by~$s(x_1) = 1$ and~$s(x_2) = s(x_3) = 0$, its representative clause is~$(\neg x_1 \vee x_2 \vee x_3)$. 

Observe that this representation ensures that a maxterm~$s$ is the unique assignment to~$\X'$ that does not satisfy its representative clause. Then, it is easily verified that a constraint~$f$ is equivalent to the conjunction of its maxterms representatives. This insight complements the following known characterization of weakly positive and weakly negative constraints.
\begin{lemma}[{\cite[Lemma 4.20]{KhannaSTW00}}]
    A constraint~$f$ is weakly positive (resp.~negative) if and only if all its maxterms are weakly positive (resp.~negative).
\end{lemma}
Now, to aid us in giving the required essential implementations, we start by showing the following.
\begin{lemma} \label{lem:wp-wn-imps-constants}
    Let~$\F$ be a constraint set that is either weakly positive or weakly negative. If~$\F$ is neither~$0$-valid nor~$1$-valid, then~$\F \imp \{\true, \false\}$.
\end{lemma}
\begin{proof}
    We give the proof for a weakly positive constraint set~$\F$; the proof for weakly negative constraint sets is analogous. Since~$\F$ is not~$1$-valid, there must be some constraint~$f_1 \in \F$ over variables~$x_1, \ldots, x_p$ that, when written as conjunction of its maxterm representatives, contains a clause of the form~$\false(x_i)$ for some~$i \in [p]$. Then, we claim that the constraint~$f_1'(x_i) := \exists \{x_1, \ldots, x_p\} \setminus x_i \mathrm{~s.t.~} f_1(x_1, \ldots, x_p)$ must be equivalent to the constraint~$\false(x_i)$. To see this, note that no assignment that sets~$x_i$ to~$1$ can satisfy~$f_1$ as~$f_1$ contains~$\false(x_i)$ as a clause. So,~$f_1(1)=0$. Secondly, since we assume throughout the paper to only consider satisfiable constraints, there must be a satisfying assignment for~$f_1$ that sets~$x_i$ to~$0$. So,~$f_1(0)=1$. As such,~$f_1'$ is equivalent to the constraint~$\false$, and by \cref{obs:quantified-implementation},~$f_1 \imp \false$.

    Since~$\F$ is not~$0$-valid, there must be a constraint~$f_2 \in \F$ over variables~$y_1, \ldots, y_q$ that, when written as conjunction of its maxterm representatives, contains a clause of the form~$(y_1 \vee \ldots \vee y_r)$ for some~$r \leq q$. Consider the constraint~$f_2'(y_1) := \exists y_{r+1}, \ldots, y_q \mathrm{~s.t.~} f_2\left[ \stack{y_2}{0} \comma \ldots \comma \stack{y_r}{0} \right]$. By \cref{obs:quantified-implementation,lem:implementation-with-constants},~$\{f_2, \false\} \imp f_2'$. Clearly,~$f_2'(0) = 0$. However, if~$f_2'(1)$ would also be~$0$, then~$(y_1 \vee \ldots y_r)$ would not be a maxterm representative of~$f$, since assigning~$0$ to just the variables~$y_2, \ldots, y_r$ would already fix~$f_2$ to~$0$. Hence,~$f_2'(1)=1$, making~$f_2'$ equivalent to the constraint~$\true$. We conclude that~$\F \imp\{\true, \false\}$.
\end{proof}
Now, we can follow a chain of implementations as given by Khanna et al.\ to show the following.
\begin{lemma} \label{lem:wp-imps-or3-1}
    Let~$\F$ be a constraint set that is not IHS-B. If~$\F$ is weakly positive, then~$\{\F, \true, \false\} \imp \orr_{3,1}$. If~$\F$ is weakly negative, then~$\{\F, \true, \false\} \imp \orr_{3,2}$.
\end{lemma}
\begin{proof}
    We prove the claim for weakly positive~$\F$; the claim for weakly negative constraint sets follows analogously. The proof of this lemma follows almost directly from the proof of Lemma 7.18 in the previous work by Khanna et al.\ \cite{KhannaSTW00}. This lemma makes the same claim but using a different definition for implementations. Following all steps of their proof, one can see that almost all of their implementations also happen to be essential implementations. 
    
    One can verify this by first noting that some of these implementations consist only of existentially quantifying over some variables in a constraint and possibly fixing some of them to~$0$. By \cref{obs:quantified-implementation,lem:implementation-with-constants}, these are essential implementations. For almost all other implementations, it is easy to verify that they are essential implementations when considering the first mentioned constraint of the implementation as its head.

    The proof by Khanna et al.\ contains just one implementation that itself is not an essential one, which is an implementation of~$\eqq$ using~$\orr_{2,1}$. We complete our proof by giving an alternative implementation that still shows~$\orr_{2,1} \imp \eqq$. In particular, we argue that~$f:=(\neg v_1 \vee v_2) \wedge (\neg x \vee v_1) \wedge (\neg y \vee v_1) \wedge (\neg v_2 \vee x) \wedge (\neg v_2 \vee y)$ is an essential implementation of~$(x = y)$, with dummy variables~$v_1$ and~$v_2$ and the first constraint application as head.

    To show that this implementation satisfies Property~\textbf{(I\ref{itm:I1})}, we show that an assignment~$s$ to~$\{x,y\}$ can be extended to a satisfying assignment for~$f$ if and only if~$s(x) = s(y)$. For one direction, suppose that an assignment~$s$ with~$s(x) = s(y)$ is given. This can be extended to a satisfying assignment for~$f$ by assigning~$v_1$ and~$v_2$ the same value as~$x$ and~$y$ as well. 
    
    For the other direction, suppose that a satisfying assignment~$s'$ to~$f$ is given and assume w.l.o.g.\ that~$s(x) = 1$. Then, to satisfy the second constraint application of~$f$, it must satisfy~$s(v_1) = 1$. To satisfy the first constraint application of~$f$, the assignment must satisfy~$s(v_2) = 1$, so to satisfy the fifth constraint application, it must satisfy~$s(y) = 1$. Thus,~$s(y) = s(x)$.

    This establishes Property~\textbf{(I\ref{itm:I1})} to hold. To see that Property~\textbf{(I\ref{itm:I2})} is satisfied, we note that any assignment to~$x$ and~$y$ can be extended to a satisfying assignment for the last four constraint applications of~$f$ by setting~$v_1$ to~$1$ and~$v_2$ to~$0$.
\end{proof}
Now, the final lemma of this section follows rather quickly.
\begin{lemma}\label{lem:wp-implementation}
    Let~$\F$ be a constraint set that is neither~$0$-valid, nor~$1$-valid nor IHS-B. If~$\F$ is weakly positive, then~$\F \imp \{\orr_{3,1}, \orr_{2,1}, \true, \false\}$. If~$\F$ is weakly negative, then~$\F \imp \{\orr_{3,2}, \orr_{2,1}, \true, \false\}$.
\end{lemma}
\begin{proof}
    For the first part of the lemma, let~$\F$ be weakly positive.
    By \cref{lem:wp-wn-imps-constants},~$\F \imp \{\true, \false\}$, so, by \cref{lem:wp-imps-or3-1} and transitivity of essential implementations,~$\F \imp \orr_{3,1}$. It remains to see that~$\F \imp \orr_{2,1}$, but this is easily verified as~$(\neg x \vee y \vee y)$ is an essential implementation of~$(\neg x \vee y)$.

    If~$\F$ is weakly negative, we can use the same lemma statements to show analogously that~$\F \imp \{\orr_{3,2}, \orr_{2,1}, \true, \false\}$.
\end{proof}
Now, \cref{lem:wp-canonical-hardness,lem:wp-implementation} combine to prove Case~\ref{itm:dichotomy-4} of \cref{thm:dichotomy}.

\subsection{All other constraint sets} \label{ssec:other}
In this section, we prove Case~\ref{itm:dichotomy-5} of \cref{thm:dichotomy}. For constraint sets~$\F$ that do not belong to any of the categories specified in the first four cases of \cref{thm:dichotomy}, we pick~$\{\nae_5, \neqq, \eqq\}$ as the canonical constraint set. We prove in \cref{sssec:appendix-other-canonical-hardness} that~\textsc{$c$-essential Detection for \mincsp{$\{\nae_5, \neqq, \eqq\}$}} is NP-hard for every~$c \in \Rone$, even on formulas~$F$ with~$\opt(F) = 1$. Then, in \cref{sssec:appendix-other-implementation}, we show that if~$\F$ is any constraint set that is neither $0$-valid, nor $1$-valid, nor bijunctive, nor affine, nor weakly positive, nor weakly negative, then $\F \imp \{ \nae_5, \neqq, \eqq \}$.

\subsubsection{Hardness of c-essential detection for the canonical problem} 
\label{sssec:appendix-other-canonical-hardness}
Using a reduction from the $3$-SAT problem, we prove the following result.

\newcommand{\nx}{\overline{x}}
\newcommand{\var}{\text{var}}
\begin{lemma}
    Unless P=NP, there is no constant~$c \in \Rone$ for which a polynomial-time $c$-essential detection algorithm exists for~\mincsp{$\{ \nae_5, \neqq, \eqq \}$}. Not even for instances in which the optimal solution is of size at most~$1$.
\end{lemma}
\begin{proof}
    For a given constant~$c \in \Rone$, we provide a reduction from $3$-SAT. First, we describe how to transform a $3$-SAT formula~$F$ over variables~$\X$ into a weighted~$\{ \nae_5, \neqq, \eqq \}$-formula~$F'$ over variables $\X'$.
    \begin{itemize}
        \item For every variable~$x \in \X$, we add variables~$x$ and~$\nx$ to~$\X'$, and, to enforce that~$\nx = \neg x$ in any satisfying assignment of~$F'$, we add a constraint application~$(x \neq \nx)$ with weight~$c+1$ to~$F'$.
        \item We add variables~$w_1$ and~$w_2$ to~$\X'$ and we add the constraint application~$(w_1 = w_2)$ with weight~$1$ to~$F'$.
        \item For a given literal $\ell$ occurring in some clause of~$F$, we define $\text{var}(\ell)$ to be the corresponding variable in~$\X'$. That is: if~$\ell$ is a positive variable~$x$, then~$\var(\ell) = x$, and if~$\ell$ is a negated variable~$\neg x$, then~$\var(\ell) = \nx$. Now, for every clause~$(\ell_1 \vee \ell_2 \vee \ell_3)$ in~$F$, we add a constraint application~$\nae(w_1, w_2, \var(\ell_1), \var(\ell_2), \var(\ell_3))$ with weight~$c+1$ to~$F'$.
    \end{itemize}
    We continue by proving that the output of a~$c$-essential detection algorithm on~$F'$ can be used to determine whether~$F$ is satisfiable. To do so, we first prove the following claim.
    \begin{claim} \label{clm:nae-5-reduction-satisfiability}
        $F$ is satisfiable iff~$F'$ is satisfiable.
    \end{claim}
    \begin{claimproof}
        For one direction, suppose that~$F$ is satisfiable and let~$s$ be an assignment to~$\X$ that satisfies~$F$. We show how this can be extended to an assignment~$s'$ to~$\X'$ that satisfies~$F'$. First, for every $x \in \X$, we set~$s'(x) = s(x)$ and~$s'(\nx) = \neg s(x)$, thereby satisfying all constraint applications of the form~$(x \neq \nx)$. Furthermore, since every clause~$(\ell_1 \vee \ell_2 \vee \ell_3)$ has at least one of its literals assigned to~$1$ by~$s$, at least one of the variables $\var(\ell_1), \var(\ell_2), \var(\ell_3)$ is assigned~$1$ by~$s'$. Next, we set~$s'(w_1) = s(w_2) = 0$, thereby satisfying all constraint applications of the form~$\nae(w_1, w_2, \var(\ell_1), \var(\ell_2), \var(\ell_3))$ and the constraint application~$(w_1 = w_2)$. Thus,~$s'$ satisfies~$F'$.

        For the other direction, suppose that~$F'$ is satisfiable and let~$s'$ be an assignment to~$\X'$ that satisfies~$F'$. Since~$s'$ satisfies the constraint~$(w_1 = w_2)$ in particular, we either have~$s'(w_1) = s'(w_2) = 0$ or~$s'(w_1) = s'(w_2) = 1$. We may assume w.l.o.g.\ that the former is true: if the latter is true, the complement of~$s'$ would be an assignment with~$s'(w_1) = s'(w_2) = 0$ that also satisfies~$F'$, since all constraint applications in~$F'$ are complementive. (Recall the definition of complementive constraints from \cref{sec:preliminaries}.) Now, we define~$s$ to be the restriction of~$s'$ to~$\X$ and show that~$s$ satisfies~$F$.
        
        Since~$s'$ assigns~$w_1$ and~$w_2$ a value of~$0$, while satisfying all constraint applications that are of the form $\nae(w_1, w_2, \var(\ell_1), \var(\ell_2), \var(\ell_3))$, it follows that for all such constraint applications, at least one of the variables~$\var(\ell_1), \var(\ell_2), \var(\ell_3)$ is set to~$1$ by~$s'$. Thus, at least one of the literals~$\ell_1, \ell_2, \ell_3$ equals~$1$ in the assignments~$s'$ and~$s$. Hence,~$s$ satisfies all clauses of~$F$.
    \end{claimproof}

    Now, suppose that~$S'$ is the output of a $c$-essential detection algorithm with~$k=1$ on~$F'$. We use the above claim to show that~$S'$ contains the constraint application~$(w_1 = w_2)$ if and only if~$F$ is unsatisfiable, thereby allowing us to determine the satisfiability of~$F$ by running a $c$-essential detection algorithm on~$F'$. We use the remainder of the proof to prove the two directions of this claim separately.

    For the first direction, suppose that~$F$ is satisfiable. By \cref{clm:nae-5-reduction-satisfiability},~$F'$ is also satisfiable. Hence, the unique optimal solution to the \wmincsp{$\{ \nae_5, \neqq, \eqq \}$}-instance~$F'$ is simply the empty set of size~$0$. Since~$0 \leq k = 1$,~$S$ has to satisfy Property~\textbf{(G\ref{itm:G1})} and must be a subset of some optimal solution. Therefore,~$S$ is the empty set and does not contain the constraint application~$(w_1 = w_2)$.

    For the other direction, suppose that~$F$ is not satisfiable. By \cref{clm:nae-5-reduction-satisfiability},~$F'$ is also not satisfiable. However, note that~$F'-\{(w_1 = w_2)\}$ is satisfiable: any arbitrary assignment~$s$ to the variables in~$\X$ can be extended to a satisfying assignment for~$F'-\{(w_1 = w_2)\}$ by setting~$x'$ to~$\neg s(x)$ for every~$x \in \X$ and setting~$w_1$ to~$0$ and~$w_2$ to~$1$. Since every constraint application besides the weight-$1$ constraint application~$(w_1 = w_2)$ has weight~$c+1$, every solution without~$(w_1 = w_2)$ has a weight of at least~$c+1$. Hence,~$\{(w_1 = w_2)\}$ is the unique optimal solution to~$F'$ and~$(w_1 = w_2)$ is even $c$-essential. Since~$k = 1$ equals the value of an optimal solution to~$F_2$,~$S$ must satisfy Property~\textbf{(G\ref{itm:G2})} and contain all $c$-essential constraint applications. In particular, it must contain the constraint application~$(w_1 = w_2)$.

    This shows that the satisfiability of~$F$ can be determined from the output of a $c$-essential detection algorithm on~$F'$. Since it is NP-hard to determine the satisfiability of~$F$, this shows that no polynomial-time~$c$-essential detection algorithm exists for \wmincsp{$\{ \nae_5, \neqq, \eqq \}$} unless P=NP. By \cref{lem:unweighted-to-weighted}, the same holds for the unweighted variant of the problem. Moreover, the value of an optimal solution to~$F'$ is at most~$1$ in both of the cases above, so this even holds for the restriction of \mincsp{$\{ \nae_5, \neqq, \eqq \}$} where optimal solutions are restricted to be of size at most~$1$.
\end{proof}

\subsubsection{Essential implementation of the canonical constraint set}
\label{sssec:appendix-other-implementation}
In this section, we show that every constraint set that is neither $0$-valid, nor $1$-valid, nor bijunctive, nor affine, nor weakly positive, nor weakly negative provides an essential implementation of~$\{ \nae_5, \neqq, \eqq \}$. First, we show that it suffices to prove that~$\nae_3$ and~$\neqq$ are implemented, by proving \cref{lem:neq-implements-eq,lem:nae3-implements-nae5} below.
\begin{lemma} \label{lem:neq-implements-eq}
    $\neqq \imp \eqq$.
\end{lemma}
\begin{proof}
    We show that~$(x \neq v) \wedge (v \neq y)$, using dummy variable~$v$ and head~$(x \neq v)$, is an essential implementation of $(x = y)$. Clearly,~$(x \neq v) \wedge (v \neq y)$ is only satisfiable if~$x=y$, which proves Property~\textbf{(I\ref{itm:I1})}. Also note that for any assignment to~$x$ and~$y$, there is an assignment to~$v$ such that $(v \neq y)$ is satisfied, so Property~\textbf{(I\ref{itm:I2})} is satisfied.
\end{proof}
\begin{lemma} \label{lem:nae3-implements-nae5}
    $\{ \nae_3, \neqq \} \imp \nae_5$.
\end{lemma}
\begin{proof}
    We show this in two parts, which combine to show the lemma due to the transitivity of essential implementations (\cref{lem:transitivity}).

    \begin{claim}
        $\{ \nae_3, \neqq \} \imp \nae_4$.
    \end{claim}
    \begin{claimproof}
        We show that $\nae(a, b, v_1) \wedge (v_1 \neq v_2) \wedge \nae(v_2, c, d)$, using dummy variables~$v_1, v_2$ is an essential implementation of~$\nae(a,b,c,d)$ with the first constraint application as head.

        For one direction of Property~\textbf{(I\ref{itm:I1})}, suppose that~$\nae(a,b,c,d)$ is satisfied. We distinguish two cases.
        \begin{itemize}
            \item Suppose that~$a = b$ and~$c = d$. By assumption, these four variables are not all equal, so we must have~$a \neq c$. Hence, if we set~$v_1 = \neg a$ (thereby satisfying the first constraint application), and we set~$v_2 = \neg v_1$ (thereby satisfying the second constraint application), we get that~$v_2 = a \neq c$, thereby also satisfying the third constraint application.
            \item Otherwise, we have that~$a \neq b$ or~$c \neq d$. Assume w.l.o.g.\ that at least the former holds. Then, the first constraint application is satisfied. To satisfy the remaining two constraint application, we set~$v_1 = c$ and~$v_2 = \neg c$.
        \end{itemize}
        For the other direction of Property~\textbf{(I\ref{itm:I1})}, suppose that~$\nae(a,b,c,d)$ is not satisfied. This means that $a = b = c = d$. Then, to satisfy the first and third constraint application respectively, $v_1$ and $v_2$ must both equal $\neg a$, which means that the second constraint application is not satisfied. So indeed, there is no assignment to the dummy variables $v_1, v_2$ that satisfies all three constraint applications. Hence, an assignment to~$\{a,b,c,d\}$ satisfies~$\nae(a,b,c,d)$ if and only if it can be extended to a satisfying assignment for~$\nae(a, b, v_1) \wedge (v_1 \neq v_2) \wedge \nae(v_2, c, d)$. As such, Property~\textbf{(I\ref{itm:I1})} holds.

        To show that Property~\textbf{(I\ref{itm:I2})} is also satisfied, we need to show that every assignment to~$\{a,b,c,d\}$ can be extended to a satisfying assignment for~$(v_1 \neq v_2) \wedge \nae(v_2, c, d)$. For an arbitrary assignment to the original variables, this can be achieved by setting~$v_1$ equal to~$c$ and~$v_2$ to~$\neg c$.
    \end{claimproof}

    A similar argument yields that $\{ \nae_4, \nae_3, \neqq \} \imp \nae_5$, as $\nae(a,b,c,v_1) \wedge (v_1 \neq v_2) \wedge \nae(v_2, d, e)$ with dummy variables~$v_1, v_2$ is an essential implementation of~$\nae(a, b, c, d, e)$, with the first constraint application as head. This proves the lemma.
\end{proof}
We continue to prove several lemmas that eventually culminate in a proof (of \cref{lem:not-in-any-category-imps-nae5}) that, under the right assumptions,~$\F \imp \{\nae_3, \neqq\}$. We start by following a sequence of implementations that is very similar to one given by Schaefer~\cite{Schaefer78}. In our proofs, we show that Schaefer's implementations are either already essential or, when they are not, provide alternative ones that are essential.

The first lemma below presents an insight that will lead to a useful case distinction in the eventual proof.
\begin{lemma}[Based on {\cite[Lemma 4.3]{Schaefer78}}] \label{lem:not-0-valid-1-valid}
    If~$\F$ is a constraint set that is neither $0$-valid nor $1$-valid, then one of the following must hold.
    \begin{enumerate}[(i)]
        \item $\F \imp \{\true, \false\}$ \label{itm:hardness-when-true-false};
        \item $\F \imp \neqq$ and every $f \in \F$ is complementive. \label{itm:hardness-when-complementive}
    \end{enumerate}
\end{lemma}
\begin{proof}
    We distinguish two cases.
    \begin{itemize}
        \item \textbf{Case 1: every~$f \in \F$ is $0$-valid or $1$-valid.} Since~$\F$ is neither $0$-valid nor $1$-valid, there must be both an~$f_0 \in \F$ that is $0$-valid but not $1$-valid and some~$f_1 \in \F$ that is $1$-valid but not $0$-valid. Clearly, since~$f_0$ is $0$-valid,~$f_0(x, \ldots, x)$ is an essential implementation of~$\false(x)$. Likewise,~$f_1(x, \ldots, x)$ is an essential implementation of~$\true(x)$. As such,~$\F \imp \{ \true, \false \}$.

        \item \textbf{Case 2: There is some~$f \in \F$ that is neither $0$-valid nor $1$-valid.} Let~$s$ be a satisfying assignment for~$f$ and define $f'(x, y) := f \left[ \stack{s\inv(0)}{x} \comma \stack{s\inv(1)}{y}  \right]$. Since~$s$ is a satisfying assignment for~$f$, we find that~$f'(0,1) = 1$. Since~$f$ is neither $0$-valid nor $1$-valid, we find that $f'(0,0) = f'(1,1) = 0$. Thus, depending on the value of~$f'(1,0)$,~$f'$ is either equivalent to~$(\neg x \wedge y)$ or to~$(x \neq y)$. 
        
        In the former case, note that~$f'(x,y)$ with~$y$ treated as dummy variable is an essential implementation of~$\false(x)$ and that~$f'(x,y)$ with~$x$ treated as dummy variable is an essential implementation of~$\true(y)$. Thus, $f \imp\{\true, \false\}$. In the latter case, we have that~$f \imp\neqq$.
    \end{itemize}
    Suppose now that $\F$ does not essentially implement $\{\true, \false\}$. We know by the above case distinction that~$\F \imp \neqq$ and it remains to prove that every constraint in~$\F$ is complementive. Before doing so, note that~$\F$ also does not essentially implement just the constraint $\true$. If it did, it would also essentially implement $\false$, as~$\true(x) \wedge (x \neq y)$ is an essential implementation of~$\false(y)$. Likewise,~$\F$ does not essentially implement~$\false$.
    
    Now, suppose for contradiction that~$f$ is a constraint in~$\F$ over~$\X$ that is not complementive. Let~$s$ be such that it satisfies~$f$ but its complement~$\overline{s}$ does not. If~$s$ is~$K_0$ (the all-zero assignment) or~$K_1$ (the all-one assignment), then ~$f'(x):=f\left[ \stack{\X}{x} \right]$ must be equivalent to~$\true(x)$ or~$\false(x)$, as the complement of~$s$ does not satisfy~$f$. Since~$f \imp f'$ by \cref{obs:substituted-implementation}, this contradicts that~$f$ essentially implements neither~$\true$ nor~$\false$.
    
    As such, we assume from now on that~$s$ is neither~$K_0$, nor~$K_1$. Then, if we define $f'(x,y):=f\left[ \stack{s\inv(0)}{x} \comma \stack{s\inv(1)}{y} \right]$, we know that~$f'(0,1) = 1$ and~$f'(1,0) = 0$.

    We can see that~$f'(0,0) = 1$, since~$f'(x,y)$ with~$x$ treated as dummy variable would be an essential implementation of~$\true(y)$ if~$f'(0,0)=0$. This is not possible as we assumed $\true$ not to be essentially implemented by~$\F$. Likewise, we find that~$f'(1,1)=1$, since~$f'(x,y)$ with~$y$ treated as dummy variable would be an essential implementation of~$\false(x)$ if~$f'(1,1) = 0$. Thus,~$f'(x,y)$ is equivalent to~$(\neg x \vee y)$.

    However, because~$\F \imp \neqq$, we still find that~$\F \imp \true$, since $(x \vee \neg v) \wedge (x \neq v)$ is an essential implementation of~$\true(x)$ with dummy variable~$v$ and head~$(x \vee \neg v)$. This is a contradiction, so we conclude that~$\F$ is complementive if~$\F$ does not essentially implement $\{\true, \false\}$.
\end{proof}
To continue, we first focus on the case that~$\F \imp \{\true, \false\}$. In this setting, the following result provides a good starting point.
\begin{lemma}[Based on {\cite[Lemma 3.2]{Schaefer78}}] \label{lem:not-wp}
    Let $f$ be a Boolean constraint. If it is not weakly negative, then at least one of the following is true: $\{f, \true, \false \} \imp \neqq$ or $\{f, \true, \false \} \imp \orr_2$. If it is not weakly positive, then at least one of the following is true: $\{f, \true, \false \} \imp \neqq$ or $\{f, \true, \false \} \imp \nor_2$.
\end{lemma}
\begin{proof}
    We prove the statement that assumes that $f$ is not weakly negative. The proof for the second statement follows analogously. It follows from \cref{lem:wpos-wneg-characterization} that there is some set of variables~$V$ that is $1$-consistent and $1$-closed for~$f$ but such that~$K_{1,V}$ does not satisfy~$f$. (Recall the definition of~$K_{1,V}$ from \cref{sec:preliminaries} as the assignment that assigns variables to~$1$ if and only if they are in~$V$). If we let~$\X$ denote the full set of variables that~$f$ is defined over, we can define~$\overline{V} := \X \setminus V$. Now, let~$W$ be a largest subset of~$\overline{V}$ such that $K_{0,W}$ satisfies~$f$. Note that~$|W| \geq 1$ because~$V$ is $1$-closed and that~$|W| <|\overline{V}|$ because~$K_{1,V}$ does not satisfy~$f$.

    Next, let~$z \in \overline{V} \setminus W$ and let~$s$ be a satisfying assignment for~$f$ such that~$s(z) = 0$ and $s(v) = 1$ for all~$v \in V$. Note that such an assignment exists because~$V$ is $1$-consistent and $1$-closed.

    Now, let~$W_0$ and~$W_1$ be the subsets of~$W$ that are assigned values of~$0$ and~$1$ by~$s$ respectively. Similarly, let~$\overline{W}_0$ and~$\overline{W}_1$ be the subsets of~$\X \setminus W$ that are assigned values of~$0$ and~$1$ by~$s$ respectively. Note that~$z \in \overline{W}_0$.

    Using the notation defined so far, we can define $f'(x,y) := f\left[ \stack{W_0}{0} \comma \stack{W_1}{x} \comma \stack{\overline{W}_0}{y} \comma \stack{\overline{W}_1}{1} \right]$. First note that $\{f, \true, \false\} \imp f'$ by \cref{lem:implementation-with-constants} and by transitivity of essential implementations. Next note that $f'(0,1) = f'(1,0) = 1$ because~$K_{0,W}$ and~$s$ satisfy~$f$. Also, by maximality of~$W$,~$f'(0,0) = 0$. Thus,~$f'$ is equivalent to either~$\orr_2$ or~$\neqq$ depending on the value of $f'(1,1)$.
\end{proof}
Thus, if~$\F$ is neither weakly positive nor weakly negative, then~$\{\F, \true, \false\}$ essentially implements~$\neqq$ or~$\{\orr_2, \nor_2\}$. This makes the following result quite useful.
\begin{lemma}[Based on {\cite[Negated Substitution Lemma]{Schaefer78}}] \label{lem:negated-substitution-lemma}
    Let~$f$ be a Boolean constraint over variables $\{x_1, \ldots, x_n\}$, let~$i \in [n]$, and let~$\fneq$ either be~$\{\neqq\}$ or~$\{\orr_2, \nor_2\}$. Then,~$\fneq \cup\{f\} \imp f\left[ \stack{x_i}{\neg x_i} \right]$.
\end{lemma}
\begin{proof}
    It is easy to verify that $f(x_1, \ldots, x_{i-1}, v,x_{i+1}, \ldots, x_n) \wedge(v \neq x_i)$ is an essential implementation of~$f\left[ \stack{x_i}{\neg x_i} \right]$ with dummy variable~$v$ and the first constraint application as head. Likewise, because~$(x_i \neq v)$ is satisfied if and only if~$(v \vee x_i) \wedge (\neg v \vee \neg x_i)$ is satisfied, it is easy to see that $f(x_1, \ldots, x_{i-1}, v,x_{i+1}, \ldots, x_n) \wedge (v \vee x_i) \wedge (\neg v \vee \neg x_i)$ is an essential implementation of~$f\left[ \stack{x_i}{\neg x_i} \right]$ with dummy variable~$v$ and the first constraint application as head.
\end{proof}
The above result will prove to be useful several times in the remainder of this section and we start by using it to prove the following.
\begin{lemma} \label{lem:not-wp-imps-neq}
    Let~$\F$ be a constraint set that is neither weakly positive, nor weakly negative. Then,~$\F \cup \{\true, \false\} \imp \neqq$.
\end{lemma}
\begin{proof}
    \cref{lem:not-wp} implies that~$\F \cup \{\true, \false\} \imp \neqq$ or~$\F \cup \{\true, \false\} \imp \{\orr_2, \nor_2\}$. For the former case, this immediately shows the lemma, and for the latter case, we show that~$\{\orr_2, \nor_2\} \imp \neqq$. Since~$\orr_{2,1}$ is obtained by negating one of the two variables in an~$\orr_2$~constraint, it follows from \cref{lem:negated-substitution-lemma} that~$\{\orr_2, \nor_2\} \imp \orr_{2,1}$. Thus, it suffices to show that~$\{\orr_2, \nor_2, \nor_{2,1}\} \imp \neqq$.
    
    We show that $f:=(\neg v_1 \vee v_2) \wedge (\neg x \vee v_1) \wedge (y \vee v_1) \wedge(x \vee \neg v_2) \wedge (\neg y \vee \neg v_2)$ is an essential implementation of~$(x \neq y)$ with dummy variables~$v_1$, and~$v_2$ and the first constraint application as head. To see that $f$ is satisfiable if and only if~$s(x) \neq s(y)$, we distinguish two cases. 
    
    First, consider an assignment~$s \colon \{x,y,v_1,v_2\} \rightarrow \{0,1\}$ in which~$s(x) = 1$. For any such assignment to satisfy the second constraint application of~$f$, it must satisfy~$s(v_1) = 1$. Then, to satisfy the first constraint application of~$f$, it must satisfy~$s(v_2) = 1$, so to satisfy the fifth constraint application, it must satisfy~$s(y) = 0$. This assignment to the four variables satisfies all constraint applications of~$f$ (and is thus the only satisfying assignment of~$f$ in which~$x$ is assigned~$1$).
    
    For the second case, consider an assignment~$s \colon \{x,y,v_1,v_2\} \rightarrow \{0,1\}$ in which~$s(x) = 0$. For this assignment to satisfy the fourth constraint application of~$f$, it must satisfy~$s(v_2) = 0$. Then, to satisfy the first constraint application of~$f$, it must satisfy~$s(v_1) = 0$, and to satisfy the third constraint application of~$f$, it must satisfy~$s(y) = 1$. This assignment to the four variables satisfies all constraint applications of~$f$  (and is thus the only satisfying assignment of~$f$ in which~$x$ is assigned~$0$).
    
    Thus,~$f$ is satisfiable if and only if~$s(x) \neq s(y)$, which shows Property~\textbf{(I\ref{itm:I1})} to hold. To see that Property~\textbf{(I\ref{itm:I2})} is satisfied, note that any assignment to~$x$ and~$y$ can be extended to a satisfying assignment for the last four constraint applications of~$f$ by setting~$v_1$ to~$1$ and~$v_2$ to~$0$.
\end{proof}
So far, we have used the assumptions that constraints are not $0$-valid, not $1$-valid, not weakly positive, and not weakly negative to show that certain essential implementations must exist. We continue by doing the same for the assumption that a constraint is not affine.
\begin{lemma}[Based on {\cite[Lemma 3.3]{Schaefer78}}] \label{lem:not-affine}
    Let~$f$ be a Boolean constraint. If~$f$ is not affine, then there is some Boolean constraint~$g$ over four variables such that~$\{f, \neqq\} \imp g$, and~$g(0,0,0,0) = g(0,0,1,1) = g(0,1,0,1) = 1$, and~$g(0,1,1,0) = 0$.
\end{lemma}
\begin{proof}    
    Since~$f$ is not affine, it follows from \cref{lem:affine-characterization} that there are three satisfying assignments~$s_0, s_1, s_2$ for~$f$ such that~$s_0 \oplus s_1 \oplus s_2$ does not satisfy~$f$. Obtain~$f'$ from~$f$ by negating all occurrences of variables that are assigned~$1$ in~$s_0$. By \cref{lem:negated-substitution-lemma}, $\{f, \neqq\} \imp f'$. Note that an assignment~$s$ satisfies $f'$ if and only if~$s \oplus s_0$ satisfies~$f$. As such,~$s_0 \oplus s_0 = K_0$, $s_1':=s_1 \oplus s_0$, and~$s_2' := s_2 \oplus s_0$ are all satisfying assignments of~$f'$, whereas~$s_0 \oplus s_1 \oplus s_2 \oplus s_0 = s_1 \oplus s_2 = s_1' \oplus s_2'$ is not.

    For~$i,j \in \{0,1\}$, we define~$V_{i,j}$ to be the set of variables~$x$ that satisfy~$s_1'(x) = i$ and~$s_2'(x) = j$. Then, we can define~$g(w,x,y,z) = f'\left[ \stack{V_{0,0}}{w} \comma \stack{V_{0,1}}{x} \comma \stack{V_{1,0}}{y} \comma \stack{V_{1,1}}{z} \right]$. By \cref{obs:substituted-implementation},~$g$ is essentially implemented by~$f'$ and, by transitivity, it is also implemented by $\{f, \neqq\}$. Observe that~$g(0,0,0,0) = g(0,0,1,1) = g(0,1,0,1) = 1$ because $K_0$,~$s_1'$ and~$s_2'$ satisfy~$f'$. Finally, observe that~$g(0,1,1,0)=0$ because $s_1' \oplus s_2'$ does not satisfy~$f'$.
\end{proof}
By including the constraint~$\false$, the above result can be extended as follows.
\begin{lemma}[Based on {\cite[Lemma 3.3]{Schaefer78}}] \label{lem:not-affine-uncomplmentive}
    Let~$f$ be a Boolean constraint. If it is not affine, then~$\{f, \false, \neqq\} \imp \{\orr_2, \nor_2\}$.
\end{lemma}
\begin{proof}
    Suppose, for contradiction, that~$\{f, \false, \neqq\}$ does not essentially implement $\{\orr_2, \nor_2\}$. Since~$\neqq \in \{f, \false, \neqq\}$, we know by \cref{lem:negated-substitution-lemma} that~$\{f, \false, \neqq\}$ essentially implements the entire set~$\{\orr_2, \orr_{2,1}, \nor_2\}$ if it essentially implements at least one of the constraints in that set. Thus, we obtain that~$\{f, \false, \neqq\}$ does not essentially implement any of the constraints~$\{\orr_2, \orr_{2,1}, \nor_2\}$.

    By \cref{lem:not-affine}, there is a Boolean constraint~$g$ over four variables  with~$g(0,0,0,0) = g(0,0,1,1) = g(0,1,0,1) = 1$ and~$g(0,1,1,0) = 0$ such that~$\{f, \false, \neqq\} \imp g$. By \cref{lem:implementation-with-constants},~$\{g, \false\}$ essentially implements the constraint $g'(x,y,z):= g(0,x,y,z)$, which satisfies~$g(0,0,0) = g(0,1,1) = g(1,0,1) = 1$, and~$g(1,1,0) = 0$.

    If~$g'(0,1,0) = 0$, then $\exists x \colon g'(x,y,z)$ is equivalent to $(\neg y \vee z)$. By \cref{obs:quantified-implementation},~$g'$ essentially implements this constraint and, by transitivity of essential implementations, so does~$\{f, \false, \neqq\}$. This contradicts the assumption that~$\{f, \false, \neqq\}$ does not essentially implement~$\orr_{2,1}$. Hence,~$g'(0,1,0) = 1$.

    Likewise, if~$g'(1,0,0) = 0$, then $\exists y \colon g'(x,y,z)$ is equivalent to~$(\neg x \vee z)$. By \cref{obs:quantified-implementation},~$g'$ essentially implements this constraint and, by transitivity of essential implementations, so does~$\{f, \false, \neqq\}$. This contradicts the assumption that~$\{f, \false, \neqq\}$ does not essentially implement~$\orr_{2,1}$. Hence,~$g'(1,0,0) = 1$.

    Now, we find that~$g'(x,y,0)$ is equivalent to~$(\neg x \vee \neg y)$. Since $\{f, \false, \neqq\}$ essentially implements this constraint, this contradicts the assumption that~$\{f, \false, \neqq\}$ does not essentially implement~$\nor_{2}$. Thus, we conclude that~$\{f, \false, \neqq\} \imp \{\orr_2, \nor_2\}$.
\end{proof}
Now, the last few results can be combined into the following statement.
\begin{lemma}\label{lem:uncomplementive-imps-or+nor}
    Let~$\F$ be a constraint set that is neither weakly positive, nor weakly negative, nor affine. Then,~$\F \cup \{\true, \false\} \imp \{\orr_2, \nor_2, \neqq\}$.
\end{lemma}
\begin{proof}
    By \cref{lem:not-wp-imps-neq},~$\F \cup \{\true, \false \} \imp \neqq$. By \cref{lem:not-affine-uncomplmentive} we have~$\F \cup \{\false, \neqq\} \imp \{\orr_2, \nor_2\}$. By transitivity of essential implementations, we have~$\F \cup \{\true, \false\} \imp \{\orr_2, \nor_2, \neqq\}$.
\end{proof}
With the next few results, we aim to prove that a constraint set that, additionally, is not bijunctive must essentially implement the constraints~$\orr_3$ and~$\nor_3$. We start with the following lemma.
\begin{lemma}[Based on {\cite[Lemma 3.4]{Schaefer78}}] \label{lem:not-bijunctive}
    Let~$f$ be a Boolean constraint that is not bijunctive. Then, there is some Boolean constraint~$h$ over four variables such that~$\{f, \neqq\} \imp h$, and $h(0,0,0,0) = h(0,0,1,1) = h(0,1,0,1) = 1$, and $h(0,0,0,1) = 0$.
\end{lemma}
\begin{proof}
    Let~$\X$ be the set of variables that~$f$ is defined over. Since~$f$ is not bijunctive, it follows from \cref{lem:bijunctive-characterization} that there is an assignment~$s$ that satisfies~$f$ and two change sets~$U,V \subseteq \X$ for~$(f, s)$ such that~$U \cap V$ is not a change set for~$(f,s)$. We obtain~$f'$ from~$f$ by negating all occurrences of variables that are assigned~$1$ by~$s$. By \cref{lem:negated-substitution-lemma}, $\{f, \neqq\} \imp f'$. Observe now that~$K_0$ satisfies~$f'$ and that~$U$ and~$V$ are change sets for $(f', K_0)$, whereas~$U\cap V$ is not.

    Now, we define $h(w,x,y,z) = f'\left[ \stack{\X \setminus(U \cup V)}{w} \comma~~ \stack{U \setminus V}{x} \comma~~ \stack{V \setminus U}{y} \comma~~ \stack{U \cap V}{z} \right]$ and we note that~$f'\imp h$ by \cref{obs:substituted-implementation}. By transitivity of essential implementations, also~$f \imp h$. Now, observe that $h(0,0,0,0) = 1$ because~$K_0$ satisfies~$f$ and observe that $h(0,0,1,1) = h(0,1,0,1) = 1$ because~$U$ and~$V$ are change sets for~$(f', K_0)$. Finally, observe that $h(0,0,0,1)=0$ because~$U \cap V$ is not a change set for~$(f', K_0)$.
\end{proof}
By including the constraints~$\false$,~$\orr_2$, and~$\nor_2$, the above result can be extended as follows.
\newcommand{\he}{h_{\mathrm{E}}}
\newcommand{\ho}{h_{\mathrm{O}}}
\begin{lemma}[Based on {\cite[Lemmas 3.4 and 3.5]{Schaefer78}}] \label{lem:not-bijunctive-implements-or3}
    Let~$f$ be a Boolean constraint. If~$f$ is not bijunctive, then~$\{f, \neqq, \orr_2, \nor_2, \false\} \imp \{\orr_3, \nor_3\}$.
\end{lemma}
\begin{proof}
    First, we prove that~$\{f, \neqq, \orr_2, \nor_2, \false\} \imp \orr_3$. By \cref{lem:not-bijunctive}, there is a Boolean constraint~$h(w,x,y,z)$ with $h(0,0,0,0) = h(0,0,1,1) = h(0,1,0,1) = 1$ and $h(0,0,0,1) = 0$ such that $\{f, \neqq\} \imp h$. By \cref{obs:substituted-implementation,lem:negated-substitution-lemma}, $h \imp h'$ where $h'(x,y,z) := f\left[ \stack{w}{0} \comma \stack{z}{\neg z} \right]$. It follows that $h'(0,0,1) =  h'(0,1,0) = h'(1,0,0) = 1$ and that~$h'(0,0,0) = 0$.

    Now, since $(\neg x \vee \neg y) \wedge (\neg x \vee \neg z) \wedge(\neg y \vee \neg z)$ is satisfied if and only if at most one of~$x$,~$y$, and~$z$ is $1$, we have that $h'(x,y,z) \wedge (\neg x \vee \neg y) \wedge (\neg x \vee \neg z) \wedge(\neg y \vee \neg z)$ is satisfied if and only if exactly one of~$x$,~$y$, and~$z$ is $1$. We use~$\he(x,y,z)$ as shorthand for this conjunction. Although the above conjunction is not an essential implementation of~$\he$, we can use it in another implementation that is essential.

    In particular, one can verify by a case distinction that $\orr_3(x,y,z) \equiv \exists v_1, v_2, v_3, v_4, v_5, v_6: \he(x, v_1, v_4) \wedge \he(y, v_2, v_4) \wedge \he(v_1, v_2, v_5) \wedge \he(v_3, v_4,v_6) \wedge (z \vee v_3) \wedge (\neg z \vee \neg v_3)$. This shows that the above satisfies Property~\textbf{(I\ref{itm:I1})} of being an essential implementation of~$\orr_3(x,y,z)$ when treating~$v_1, \ldots, v_6$ as dummy variables. Next, we show that it also satisfies Property~\textbf{(I\ref{itm:I2})} when correctly picking which constraint application should be its head.
    
    Now, recall that the expression $\he(x, v_1, v_4)$ is written as $h'(x,v_1,v_4) \wedge (\neg x \vee \neg v_1) \wedge (\neg x \vee \neg v_4) \wedge(\neg v_1 \vee \neg v_4)$ in expanded form. Then, the above implementation of $\orr_3(x,y,z)$ is essential when treating $h'(x,v_1,v_4)$ as its head. In the previous paragraph, we established that all assignments to~$\{x,y,z\}$ that satisfy~$\orr_3(x,y,z)$ can be extended to satisfy the entire implementation. Observe that the only remaining assignment to~$\{x, y, z\}$~---~the one that assigns all three variables to~$0$ --- can be extended to a satisfying assignment of the entire implementation minus its head, by setting~$v_2$ and~$v_3$ to~$1$ and~$v_1$,~$v_4$,~$v_5$, and~$v_6$ to~$0$.

    Hence, this proves that~$\{f, \neqq, \orr_2, \nor_2, \false\} \imp \orr_3$. Next, by \cref{lem:negated-substitution-lemma}, $\{\orr_3, \neqq\} \imp \nor_3$. By transitivity, $\{f, \neqq, \orr_2, \nor_2, \false\} \imp \nor_3$ as well.
\end{proof} 
In the next two lemmas, we show that the acquired constraints~$\orr_3$,~$\nor_3$, and~$\neqq$ suffice to essentially implement~$\nae_3$.
\begin{lemma} \label{lem:or3+neq-imps-nae3}
    $\{\orr_3, \nor_3, \neqq\} \imp \nae_3$.
\end{lemma}
\begin{proof}
    By \cref{lem:neq-implements-eq}, $\neqq \imp \eqq$, so it suffices to prove that $\{\orr_3, \nor_3, \eqq\} \imp \nae_3$. We show that $(v = z) \wedge(x \vee y \vee v) \wedge(\neg x \vee \neg y \vee \neg v)$, with dummy variable~$v$ and head~$(v=z)$ is an essential implementation of $\nae_3$. Clearly, $(v = z) \wedge(x \vee y \vee v) \wedge(\neg x \vee \neg y \vee \neg v)$ is satisfiable if and only if $(x \vee y \vee z) \wedge (\neg x \vee \neg y \vee \neg z)$ is satisfied, which, in turn, is satisfied if and only if~$x$,~$y$, and~$z$ are not all equal. Hence, Property~\textbf{(I\ref{itm:I1})} is satisfied.

    Next, note that any assignment to~$x$,~$y$, and~$z$ extends to a satisfying assignment for $(x \vee y \vee v) \wedge(\neg x \vee \neg y \vee \neg v)$ by setting $v = \neg x$. Hence, Property~\textbf{(I\ref{itm:I2})} is satisfied.
\end{proof}
Now, the last five results combine to show the following.
\begin{lemma} \label{lem:true-false-imps-nae3+neq}
    Let~$\F$ be a constraint set that is neither weakly positive, nor weakly negative, nor affine, nor bijunctive. Then,~$\F \cup \{\true, \false\} \imp \{\nae_3, \neqq\}$.
\end{lemma}
\begin{proof}
    By \cref{lem:uncomplementive-imps-or+nor},~$\F \cup \{\true, \false\} \imp \{ \orr_2, \nor_2, \neqq\}$. Combining this with the assumption that~$\F$ contains at least one constraint that is not bijunctive, it follows from \cref{lem:not-bijunctive-implements-or3} that~$\F \cup \{\true, \false\} \imp\{\orr_3, \nor_3\}$. Then, \cref{lem:or3+neq-imps-nae3} implies that~$\F \cup \{\true, \false\} \imp \nae_3$, completing the proof.
\end{proof}
The lemma above covers all required steps to prove hardness for the first case that \cref{lem:not-0-valid-1-valid} outlines. The lemma below covers the second case. Although we re-use some of the previous lemmas to prove it, this is where our approach deviates most strongly from that of Schaefer~\cite{Schaefer78}. In his work, the second case followed easily from the first one as his reductions only needed to preserve the satisfiability between formulas. Since our work requires more control over how the solution space of formulas are related to one another, we spend a larger effort on this second case.
\begin{lemma} \label{lem:complementive-imps-nae3}
    Let $\F$ be a Boolean constraint set such that $\F$ is not affine, and $\F$ is not bijunctive, and every $f \in \F$ is complementive. Then,~$\{\F, \neqq\} \imp \nae_3$.
\end{lemma}
\begin{proof}
    We prove the claim by contradiction, so suppose that~$\{\F, \neqq\}$ does not essentially implement~$\nae_3$.

    Let~$f_1 \in \F$ be a constraint that is not affine. By \cref{lem:not-affine}, there is a constraint~$g$ over four variables such that~$\{f_1, \neqq\} \imp g$ and $g(0,0,0,0) = g(0,0,1,1) = g(0,1,0,1) = 1$, and~$g(0,1,1,0) = 0$. Now, if we define~$g'(w,x,y,z) := g(\neg w, x, y, \neg z)$ we know that $g'(1,0,0,1) = g'(1,0,1,0) = g'(1,1,0,0) = 1$ and~$g'(1,1,1,1) = 0$ and, by \cref{lem:negated-substitution-lemma}, that~$\{g, \neqq\} \imp g'$. Moreover, because $g$ --- and by extension $g'$ --- can be expressed as a conjunction of applications of the complementive constraints~$f_1$ and~$\neqq$, $g'$ is also complementive. As such, we find that $g'(0,1,1,0) = g'(0,1,0,1) = g'(0,0,1,1) = 1$ and $g'(0,0,0,0) = 0$.

    Now, suppose for contradiction that~$g'(0,0,0,1) = 0$. Then, since~$g'$ is complementive, also~$g'(1,1,1,0) = 0$. Now, we show that $g^\ast(w,x,y) := \exists z: g'(w,x,y,z)$ is equivalent to~$\nae(w,x,y)$. First note that all extensions of the assignments~$w=x=y=0$ and~$w=x=y=1$ result in assignments that do not satisfy~$g'$. Next, observe that for every other assignment to~$w$,~$x$, and~$y$ there is an assignment to~$z$ such that the resulting extension is known to satisfy~$g'$. Hence,~$g^\ast$ and~$\nae_3$ are indeed equivalent. Moreover,~$g'\imp g^\ast$ by \cref{obs:quantified-implementation}. Since~$\{f_1, \neqq\}\imp g'$ this contradicts the assumption that~$\{\F, \neqq\}$ does not essentially implement~$\nae_3$. Thus, we find that~$g'(0,0,0,1) = 1$, and, because~$g'$ is complementive, also~$g'(1,1,1,0) = 1$.

    Next, let~$f_2 \in \F$ be a constraint that is not bijunctive. Then, by \cref{lem:not-bijunctive}, there is a constraint~$h$ over four variables such that~$\{f_2, \neqq\} \imp h$ and $h(0,0,0,0) = h(0,0,1,1) = h(0,1,0,1) = 1$, and~$h(0,0,0,1) = 0$. Since~$h$ can be written as conjunction of the complementive constraints~$f_2$ and~$\neqq$, $h$ is also complementive. As such, we find that $h(1,1,1,1) = h(1,1,0,0) = h(1,0,1,0) = 1$ and~$h(1,1,1,0) = 0$.

    Now, suppose for contradiction that~$h(1,0,0,1) = 1$. Then, since~$h$ is complementive, also~$h(0,1,1,0) = 1$. We continue by showing that $g'(w,x,y,v) \wedge h(w,x,y,v)$, with dummy variable~$v$ and the first constraint application as its head, is an essential implementation of~$\nae(w,x,y)$. We start by showing that Property~\textbf{(I\ref{itm:I1})} holds for this implementation and we do so in two parts.

    First, we show that if~$w=x=y=0$ or~$w=x=y=1$ there is no~$v$ such that the extension satifies~$g'(w,x,y,v) \wedge h(w,x,y,v)$. This follows from the fact that $g'(0,0,0,0) = h(0,0,0,1) = g'(1,1,1,1) = h(1,1,1,0) = 0$.

    Secondly, we show that for all other assignments to~$w$,~$x$, and~$y$ (i.e.: all assingments in which these are not all equal) there is a~$v$ such that the resulting extension satisfies~$g'(w,x,y,v) \wedge h(w,x,y,v)$. This follows from the facts that~$g'(0,0,1,1) = h(0,0,1,1)=1$, and~$g'(0,1,0,1) = h(0,1,0,1) = 1$, and~$g'(0,1,1,0) = h(0,1,1,0) = 1$, and~$g'(1,0,0,1) = h(1,0,0,1) = 1$, and~$g'(1,0,1,0) = h(1,0,1,0) = 1$, and~$g'(1,1,0,0) = h(1,1,0,0) = 1$.

    So indeed, an assignment to~$w$,~$x$, and~$y$ satisfies~$\nae(w,x,y)$ if and only if it admits an extension that satisfies~$g'(w,x,y,v) \wedge h(w,x,y,v)$. To show that this latter conjunction is indeed an essential implementation of $\nae(x,y,z)$, it remains to show that Property~\textbf{(I\ref{itm:I2})} is satisfied for it. For this, we need to show that every assignment to to~$w$,~$x$, and~$y$ admits an extension that satisfies~$h(w,x,y,v)$. This is indeed the case as witnessed by the fact that $h(0,0,0,0) = h(0,0,1,1) = h(0,1,0,1) = h(0,1,1,0) = h(1,0,0,1) = h(1,0,1,0) = h(1,1,0,0) = h(1,1,1,1) = 1$.

    This shows that~$\{h,g'\} \imp \nae_3$. As~$\{\F, \neqq\} \imp \{h, g'\}$, this contradicts the assumption that $\{\F, \neqq\}$ does not essentially implement~$\nae_3$. Thus, we find that~$h(1,0,0,1) = 0$. Because~$h$ is complementive, also $h(0,1,1,0) = 0$. To reach our final contradiction, we use this insight to prove that $h \imp \nae_3$.

    To this end, we define~$h'(w,x,y,z) = h(w,x,y,\neg z)$ so that $h'(0,0,0,1) = h'(1,1,1,0) = h'(0,0,1,0) = h'(1,1,0,1) = h'(0,1,0,0) = h'(1,0,1,1) = 1$ and $h'(1,1,1,1) = h'(0,0,0,0) = h'(1,0,0,0) = h'(0,1,1,1) = 0$. By \cref{lem:negated-substitution-lemma},~$\{h, \neqq\} \imp h'$. Now, it is not hard to verify that $h^\ast(x,y,z) := \exists w: h'(w,x,y,z)$ is equivalent to~$\nae(x,y,z)$. By \cref{obs:quantified-implementation}, $h' \imp h^\ast$, and because $\{\F, \neqq\} \imp h'$ this contradicts that $\{\F, \neqq\}$ does not essentially implement~$\nae_3$. Thus, we conclude that this set does essentially implement~$\nae_3$.
\end{proof}
Now, the results above culminate in proving the main result of this section.
\begin{lemma} \label{lem:not-in-any-category-imps-nae5}
    Let $\F$ be a Boolean constraint set that is neither $0$-valid, nor $1$-valid, nor weakly positive, nor weakly negative, nor affine, nor bijunctive. Then, $\F \imp \{\nae_5, \neqq, \linebreak[1] \eqq\}$.
\end{lemma}
\begin{proof}
    If~$\F \imp \{\true, \false\}$, then we know by \cref{lem:true-false-imps-nae3+neq} that $\F \imp \{\nae_3, \neqq\}$. Otherwise, we obtain from \cref{lem:not-0-valid-1-valid} that every constraint~$f \in \F$ is complementive and that $\F \imp \neqq$. Then, we find that~$\F \imp \nae_3$ by \cref{lem:complementive-imps-nae3}. Thus, in any case, we have that $\F \imp \{\nae_3, \neqq\}$. Finally, by \cref{lem:neq-implements-eq,lem:nae3-implements-nae5}, we obtain that $\F \imp \{\nae_5, \neqq, \eqq\}$.
\end{proof}

\section{Conclusion and Discussion}
\label{sec:conclusion}
We extended the framework of search-space reduction via essential objects beyond the setting of graph problems for which it was originally defined. Investigating the task of detecting~$c$-essential constraint applications for \minF has resulted in a dichotomy that characterizes for which constraint sets~$\F$ this is polynomial-time solvable. We note the similarity between our dichotomy theorem and the dichotomy theorem from Bonnet et al.~\cite{BonnetEM16} regarding the constant-factor approximability of \minF in FPT time. Both dichotomies yield positive results for exactly the same classes of constraint sets. We leave it to future work to investigate whether this is a coincidence or a consequence of a deeper connection between the two tasks.

Other directions for follow-up research include extending the results of this dichotomy to CSPs beyond the Boolean domain. In fact, successfully applying the idea of~$c$-essentiality to a problem domain different from graph theory has solidified the framework as a robust method by which meaningful and non-trivial new results can be attained across problem domains. As such, one could also attempt to find results using it in even different areas such as scheduling or computational geometry.

While our dichotomy theorem characterizes when $c$-essential detection is possible for \emph{some} constant~$c$, it does not shed any light on what the smallest value of~$c$ is for which this task is tractable, nor on the relation between this value and the choice of constraint set~$\mathcal{F}$. We leave this to future work. A concrete question in this direction is to determine the smallest  constant~$c_\F$ for which \textsc{$c_\F$-essential Detection for \minF} is possible for bijunctive constraint languages~$\F$. Our approach in \cref{lem:bijunctive} of solving several cut problems in the implication graph gives a guarantee of~$c_\F = (2d^2+1)$ whenever each constraint in~$\F$ can be expressed as a 2CNF-formula with~$d$ clauses, but we do not expect this to be best-possible. Perhaps a discrete relaxation of the cut problem, like those studied in earlier work on \textsc{Almost $2$-SAT}~\cite{IwataWY16,IwataYY18}, can provide a better bound?

A final open question concerns the theoretical basis for the impossibility of \textsc{$c_\F$-essential Detection} for affine constraint languages~$\mathcal{F}$. We currently rely on the Unique Games Conjecture to rule out $c_\F$-detection algorithms for this family, via known lower bounds for \mincsp{$\{\eqq, \neqq\}$}~\cite{Khot02}. Can the same lower bound be established relative to the standard assumption P~$\neq$~NP, for example via known hardness of approximation results for \textsc{Nearest Codeword}~\cite{BhattiproluGR25}?

\bibliography{references}

@book{AroraB09,
  author       = {Sanjeev Arora and
                  Boaz Barak},
  title        = {Computational Complexity - {A} Modern Approach},
  publisher    = {Cambridge University Press},
  year         = {2009},
  url          = {http://www.cambridge.org/catalogue/catalogue.asp?isbn=9780521424264},
  isbn         = {978-0-521-42426-4},
  bibsource    = {dblp computer science bibliography, https://dblp.org},
  doi          = {10.1017/CBO9780511804090}
}

@article{BhattiproluGR25,
  author       = {Vijay Bhattiprolu and
                  Venkatesan Guruswami and
                  Xuandi Ren},
  title        = {{PCP}-free {APX}-Hardness of Nearest Codeword and Minimum Distance},
  journal      = {Electron. Colloquium Comput. Complex.},
  volume       = {{TR25}},
  eid          = {{TR25-029}},
  year         = {2025},
  url          = {https://eccc.weizmann.ac.il/report/2025/029},
  eprinttype   = {ECCC},
  eprint       = {TR25-029},
}

@article{AspvallPT79,
  author       = {Bengt Aspvall and
                  Michael F. Plass and
                  Robert Endre Tarjan},
  title        = {A Linear-Time Algorithm for Testing the Truth of Certain Quantified
                  Boolean Formulas},
  journal      = {Inf. Process. Lett.},
  volume       = {8},
  number       = {3},
  pages        = {121--123},
  year         = {1979},
  doi          = {10.1016/0020-0190(79)90002-4},
  bibsource    = {dblp computer science bibliography, https://dblp.org}
}

@article{RamanujanS17,
  author       = {M. S. Ramanujan and
                  Saket Saurabh},
  title        = {Linear-Time Parameterized Algorithms via Skew-Symmetric Multicuts},
  journal      = {{ACM} Trans. Algorithms},
  volume       = {13},
  number       = {4},
  pages        = {46:1--46:25},
  year         = {2017},
  doi          = {10.1145/3128600},
}

@inproceedings{IwataYY18,
  author       = {Yoichi Iwata and
                  Yutaro Yamaguchi and
                  Yuichi Yoshida},
  editor       = {Mikkel Thorup},
  title        = {{0/1/All CSPs}, Half-Integral {A}-Path Packing, and Linear-Time {FPT}
                  Algorithms},
  booktitle    = {59th {IEEE} Annual Symposium on Foundations of Computer Science, {FOCS}
                  2018, Paris, France, October 7-9, 2018},
  pages        = {462--473},
  publisher    = {{IEEE} Computer Society},
  year         = {2018},
  doi          = {10.1109/FOCS.2018.00051},
}

@article{IwataWY16,
  author       = {Yoichi Iwata and
                  Magnus Wahlstr{\"{o}}m and
                  Yuichi Yoshida},
  title        = {Half-integrality, LP-branching, and {FPT} Algorithms},
  journal      = {{SIAM} J. Comput.},
  volume       = {45},
  number       = {4},
  pages        = {1377--1411},
  year         = {2016},
  doi          = {10.1137/140962838},
}

@techreport{AchterbergBGRW16,
  author      = {Tobias Achterberg and Robert E. Bixby and Zonghao Gu and Edward Rothberg and Dieter Weninger},
  title       = {Presolve Reductions in Mixed Integer Programming},
  institution = {ZIB},
  address     = {Takustr.7, 14195 Berlin},
  number      = {16-44},
	url = {http://nbn-resolving.de/urn:nbn:de:0297-zib-60370},
  urn         = {urn:nbn:de:0297-zib-60370},
  year        = {2016}
}

@article{BumpusJK24,
  author       = {Benjamin Merlin Bumpus and
                  Bart M. P. Jansen and
                  Jari J. H. de Kroon},
  title        = {Search-Space Reduction via Essential Vertices},
  journal      = {{SIAM} J. Discret. Math.},
  volume       = {38},
  number       = {3},
  pages        = {2392--2415},
  year         = {2024},
  doi          = {10.1137/23M1589347},
  bibsource    = {dblp computer science bibliography, https://dblp.org}
}

@inproceedings{Schaefer78,
  author       = {Thomas J. Schaefer},
  editor       = {Richard J. Lipton and
                  Walter A. Burkhard and
                  Walter J. Savitch and
                  Emily P. Friedman and
                  Alfred V. Aho},
  title        = {The Complexity of Satisfiability Problems},
  booktitle    = {Proceedings of the 10th Annual {ACM} Symposium on Theory of Computing,
                  May 1-3, 1978, San Diego, California, {USA}},
  pages        = {216--226},
  publisher    = {{ACM}},
  year         = {1978},
  doi          = {10.1145/800133.804350},
  bibsource    = {dblp computer science bibliography, https://dblp.org}
}

@article{KhannaSTW00,
  author       = {Sanjeev Khanna and
                  Madhu Sudan and
                  Luca Trevisan and
                  David P. Williamson},
  title        = {The Approximability of Constraint Satisfaction Problems},
  journal      = {{SIAM} J. Comput.},
  volume       = {30},
  number       = {6},
  pages        = {1863--1920},
  year         = {2000},
  doi          = {10.1137/S0097539799349948},
  bibsource    = {dblp computer science bibliography, https://dblp.org}
}

@inproceedings{Khot02,
  author       = {Subhash Khot},
  editor       = {John H. Reif},
  title        = {On the power of unique 2-prover 1-round games},
  booktitle    = {Proceedings on 34th Annual {ACM} Symposium on Theory of Computing,
                  May 19-21, 2002, Montr{\'{e}}al, Qu{\'{e}}bec, Canada},
  pages        = {767--775},
  publisher    = {{ACM}},
  year         = {2002},
  doi          = {10.1145/509907.510017},
  bibsource    = {dblp computer science bibliography, https://dblp.org}
}

@article{KratschMW16,
  author       = {Stefan Kratsch and
                  D{\'{a}}niel Marx and
                  Magnus Wahlstr{\"{o}}m},
  title        = {Parameterized Complexity and Kernelizability of Max Ones and Exact
                  Ones Problems},
  journal      = {{ACM} Trans. Comput. Theory},
  volume       = {8},
  number       = {1},
  pages        = {1:1--1:28},
  year         = {2016},
  doi          = {10.1145/2858787},
  bibsource    = {dblp computer science bibliography, https://dblp.org}
}

@inproceedings{KratschW10,
  author       = {Stefan Kratsch and
                  Magnus Wahlstr{\"{o}}m},
  editor       = {Samson Abramsky and
                  Cyril Gavoille and
                  Claude Kirchner and
                  Friedhelm Meyer auf der Heide and
                  Paul G. Spirakis},
  title        = {Preprocessing of Min Ones Problems: {A} Dichotomy},
  booktitle    = {Automata, Languages and Programming, 37th International Colloquium,
                  {ICALP} 2010, Bordeaux, France, July 6-10, 2010, Proceedings, Part
                  {I}},
  series       = {Lecture Notes in Computer Science},
  volume       = {6198},
  pages        = {653--665},
  publisher    = {Springer},
  year         = {2010},
  doi          = {10.1007/978-3-642-14165-2_55},
  bibsource    = {dblp computer science bibliography, https://dblp.org}
}

@article{KimKPW25,
  author       = {Eun Jung Kim and
                  Stefan Kratsch and
                  Marcin Pilipczuk and
                  Magnus Wahlstr{\"{o}}m},
  title        = {Flow-Augmentation {III:} {C}omplexity Dichotomy for {B}oolean {CSPs} Parameterized
                  by the Number of Unsatisfied Constraints},
  journal      = {{SIAM} J. Comput.},
  volume       = {54},
  number       = {4},
  pages        = {1065--1137},
  year         = {2025},
  doi          = {10.1137/23M1553698},
  bibsource    = {dblp computer science bibliography, https://dblp.org}
}

@article{JansenV26,
  author       = {Bart M. P. Jansen and
                  Ruben F. A. Verhaegh},
  title        = {Search-space reduction via essential vertices revisited: {V}ertex multicut
                  and cograph deletion},
  journal      = {J. Comput. Syst. Sci.},
  volume       = {156},
  pages        = {103730},
  year         = {2026},
  doi          = {10.1016/J.JCSS.2025.103730},
  bibsource    = {dblp computer science bibliography, https://dblp.org}
}

@inproceedings{BonnetEM16,
  author       = {{\'{E}}douard Bonnet and
                  L{\'{a}}szl{\'{o}} Egri and
                  D{\'{a}}niel Marx},
  editor       = {Piotr Sankowski and
                  Christos D. Zaroliagis},
  title        = {Fixed-Parameter Approximability of {B}oolean {MinCSPs}},
  booktitle    = {24th Annual European Symposium on Algorithms, {ESA} 2016, Aarhus,
                  Denmark, August 22-24, 2016},
  series       = {LIPIcs},
  volume       = {57},
  pages        = {18:1--18:18},
  publisher    = {Schloss Dagstuhl - Leibniz-Zentrum f{\"{u}}r Informatik},
  year         = {2016},
  doi          = {10.4230/LIPICS.ESA.2016.18},
  bibsource    = {dblp computer science bibliography, https://dblp.org}
}

@article{FordF56, 
    title={Maximal Flow Through a Network}, 
    volume={8}, 
    DOI={10.4153/CJM-1956-045-5}, 
    journal={Canadian Journal of Mathematics}, 
    author={Ford, L. R. and Fulkerson, D. R.}, 
    year={1956}, 
    pages={399–404}
}

@book{FominLSM19,
  title={Kernelization: theory of parameterized preprocessing},
  author={Fomin, Fedor and Lokshtanov, Daniel and Saurabh, Saket and Zehavi, Meirav},
  year={2019},
  publisher={Cambridge University Press}
}

@article{KratschW20,
  author       = {Stefan Kratsch and
                  Magnus Wahlstr{\"{o}}m},
  title        = {Representative Sets and Irrelevant Vertices: {N}ew Tools for Kernelization},
  journal      = {J. {ACM}},
  volume       = {67},
  number       = {3},
  pages        = {16:1--16:50},
  year         = {2020},
  doi          = {10.1145/3390887},
  bibsource    = {dblp computer science bibliography, https://dblp.org}
}

@article{JansenW24,
  author       = {Bart M. P. Jansen and
                  Michal Wlodarczyk},
  title        = {Optimal Polynomial-Time Compression for {B}oolean Max {CSP}},
  journal      = {{ACM} Trans. Comput. Theory},
  volume       = {16},
  number       = {1},
  pages        = {4:1--4:20},
  year         = {2024},
  doi          = {10.1145/3624704},
  bibsource    = {dblp computer science bibliography, https://dblp.org}
}

@book{DFUVol7,
  title =	{The Constraint Satisfaction Problem: Complexity and Approximability},
  series =	{Dagstuhl Follow-Ups},
  ISBN =	{978-3-95977-003-3},
  ISSN =	{1868-8977},
  year =	{2017},
  volume =	{7},
  author =	{Krokhin, Andrei and Zivny, Stanislav},
  publisher =	{Schloss Dagstuhl -- Leibniz-Zentrum f{\"u}r Informatik},
  address =	{Dagstuhl, Germany},
  URN =		{urn:nbn:de:0030-drops-69752},
  doi =		{10.4230/DFU.Vol7.15301}
}

\end{document}